\documentclass[aps,pra,twocolumn,superscriptaddress]{revtex4-2}
\usepackage{bm}
\usepackage{amsmath}
\usepackage{amssymb}
\usepackage{graphicx}
\usepackage{amsfonts}
\usepackage{amsthm}
\usepackage[T1]{fontenc}
\usepackage{microtype}
\usepackage{url}

\theoremstyle{plain}
\newtheorem{theorem}{Theorem}
\newtheorem{proposition}[theorem]{Proposition}
\newtheorem{corollary}[theorem]{Corollary}

\theoremstyle{remark}
\newtheorem{remark}{Remark}

\DeclareMathOperator{\Rea}{Re}
\DeclareMathOperator{\Ima}{Im}
\DeclareMathOperator{\Tr}{Tr}
\newcommand{\xnode}{x_{\mathrm{node}}}

\begin{document}

\title{Nodal obstruction to conditioned-diffusion representations of the weak momentum}

\author{Chon-Fai Kam}
\email[Contact author: ]{dubussygauss@gmail.com}
\affiliation{Dipartimento di Fisica e Chimica ``Emilio Segr\`e'', Universit\`a degli Studi di Palermo, Via Archirafi 36, I-90123, Palermo, Italy}
\affiliation{DSIMB, Inserm, BIGR U1134, Universit\'e Paris Cit\'e \& Universit\'e de La R\'eunion, 75015, Paris, France}

\author{Kai-Wen Wong}
\affiliation{Life Actuarial Department, Taiwan Insurance Institute,\\
6th floor, No.~3, Nan-Hai Road, Taipei 100, Taiwan (R.O.C.)}

\begin{abstract}
Quantum Schr\"odinger bridge constructions establish the existence of a conditioned process
under a positivity hypothesis on the endpoint data, and three independent ones adopt it. We
show that hypothesis fails on a definite locus once the conditioned object is the two-state
amplitude $\phi^*\psi$ of a pre- and post-selected pair on configuration space, and that the
failure is testable on data already recorded. A measured fringe visibility fixes a single
length; if the contrast is limited by a zero displaced off the real axis, that length puts a
Lorentzian of order $10^{-3}$ into the current channel already reconstructed in the two-slit
trajectory experiment, and if it is limited by incoherence there is nothing there.
Whenever $\phi^*\psi$ has a moving zero of order $k$ across which probability flows, no
diffusion of constant diffusion coefficient can both carry the conditioned density and
avoid the nodal curve. The current velocity it would need grows as the inverse $2k$-th
power of the distance to that curve and points toward it on one side, placing the curve at
finite scale distance, so that it is reached and not merely approached.
Identifying the weak momentum with a bridge drift fails independently: for a free particle
post-selected in position the current velocity is exactly minus the drift of the bridge to
the target.
The positive Doob $h$-transform, and the Bernstein conditioning built on it, exist wherever
$\phi^*\psi$ is nodeless and fail on the nodal curve. There the osmotic part of the weak
momentum diverges as the inverse distance to the zero and changes sign across it with a
coefficient set by $k$, while the current part stays bounded. That boundedness is the
substantive half of the statement, resting on the signed factorization of the real-envelope
class, and it hides the obstruction from any picture built from the current velocity alone.
The results are one-dimensional, covering the separable transverse field of a two-slit
geometry and not general two-dimensional vortices. The osmotic profile is fixed by the
recorded intensity, and a same-apparatus readout in the conjugate polarization basis tests
whether the measured pointer rotation reproduces it. Nodes placed in both the preparation
and the post-selection at one position test whether the pole coefficients add, which no
single-field measurement reproduces.
\end{abstract}

\maketitle

\section{Introduction}

An expectation value of an observable is a convex combination of its eigenvalues. Measured
weakly between a preparation and a post-selection onto a different state, the same
observable can instead take a value far outside that range~\cite{aav1988}. In the two-slit
geometry the same protocol reconstructs a flow field whose integral curves look like
particle trajectories~\cite{kocsis2011}. A classical diffusion conditioned on a future outcome
exhibits both features. It carries a velocity field, and conditioning on an increasingly
unlikely outcome pushes the conditional mean outside the range the unconditioned process
would visit. A weak value might then be nothing more than ordinary conditioning of a
stochastic process on a future event. The question is old, and recent work on quantum
Schr\"odinger bridges~\cite{pavon2003,miangolarra2025,aw2026,ohzeki2026} has revived it.

Three of those constructions agree on one hypothesis. Existence
is established under a positivity condition on the endpoint data, everywhere-positive
endpoint densities in~\cite{pavon2003}, strictly positive endpoint and joint probabilities
in~\cite{miangolarra2025}, a positivity-improving channel in~\cite{aw2026}. The
restriction itself, and the program of relaxing it, go back to the
1990s~\cite{fortet1940,garbaczewski1994,blanchard1995} (Remark~\ref{rem:related}). No construction
combines the two-state amplitude with configuration space, and it is there that the
hypothesis fails on a definite locus. This paper computes that locus
(Remark~\ref{rem:related}).
We settle the question by locating the boundary of the positive conditioning, not by
contradicting a published claim. Where a conditioned diffusion exists at all, its drift and
$\Rea(p_w)/m$ are related exactly, by a sign inversion. Where the zero of $\phi^*\psi$ moves and probability flows
across it, no diffusion of constant diffusion coefficient can both carry the conditioned
density and avoid the nodal curve, the drift it would need pointing into the curve, which
is then reached with positive probability (Theorem~\ref{thm:noflux}).

Between a preparation $|\psi\rangle$ and a post-selection $|\phi\rangle$, a device coupled
weakly to an observable $A$ has its pointer displaced by the weak value
$A_w=\langle\phi|A|\psi\rangle/\langle\phi|\psi\rangle$~\cite{aav1988,aharonov1990}, the
pre- and post-selected setting itself going back to Aharonov, Bergmann, and
Lebowitz~\cite{abl1964,tsvf-review}. This quantity is complex, and its real and imaginary
parts displace two conjugate pointer variables independently~\cite{jozsa2007,dressel2014}.
For $A=\hat p$ with post-selection on position, Wiseman~\cite{wiseman2007} showed that
$\Rea(p_w)/m$ is the velocity field $\nabla S/m$ carried by the post-selected
ensemble~\cite{durr2009}, an identification whose scope has been
debated~\cite{fankhauser2021} in ways that do not bear on what follows. Optics reaches the
same object from another side.
Berry~\cite{berry2009} distinguished the optical current from the local wavevector and
noted that a weak measurement returns the latter, and Bliokh \emph{et
al.}~\cite{bliokh2013} identified the field reconstructed in the two-slit run with the real
part of the local canonical momentum of the paraxial wave, whose imaginary part is
proportional to $\nabla\ln|E|$ for a field amplitude $E$.

The computation of $\Ima(p_w)$ for a single field is not new, and the proposal to record
it in the orthogonal Stokes component is Bliokh \emph{et al.}'s~\cite{bliokh2013},
following Berry~\cite{berry2009}. This paper adds the two-state object, the conditioning that fails at its zeros, and the
discriminator that failure implies (Remark~\ref{rem:related}).

The machinery that would turn this into a stochastic account already exists. Nelson's
stochastic mechanics~\cite{nelson1966,nelson1967,nelson1985} carries a quantum state as a
diffusion with two velocity fields, a current velocity $v=\nabla S/m$ and an osmotic
velocity $u=\nu\nabla\ln\rho$ with $\nu=\hbar/2m$, the first of which is the field Wiseman
identified with $\Rea(p_w)/m$. Doob's
$h$-transform~\cite{doob1984} conditions such a diffusion on a future event, and its
two-sided version is the reciprocal process of Euclidean quantum
mechanics~\cite{zambrini1986,leonard2014}, which Schr\"odinger himself posed and read as a
classical analogue of the $\psi\bar\psi$ of wave
mechanics~\cite{schrodinger1931,schrodinger1932}. The real-time counterpart of that product
is the two-state amplitude $\phi^*\psi$, and the construction requires it to stay positive.
Readings close to the stochastic one just sketched run through Tomita's classical model of the weak
value~\cite{tomita2011,tomita2012}, the past-quantum-state
formalism~\cite{gammelmark2013}, quantum smoothing~\cite{tsang2009,chantasri2021,laverick2023},
and the time-symmetric bridges of Pavon~\cite{pavon2003} and Zambrini~\cite{zambrini1986}.
\begin{table}[b]
\caption{Where the present work sits. The last column is the hypothesis each construction
uses to establish existence.}
\label{tab:related}
\begin{ruledtabular}
\begin{tabular}{lccl}
& two-state & config.\ & positivity \\
& amplitude & space & hypothesis \\
\colrule
\cite{pavon2003} & no & yes & $\rho_0,\rho_1>0$ \\
\cite{miangolarra2025} & yes & no & $\alpha_i,\beta_j,p_{ij}>0$ \\
\cite{ohzeki2026} & yes & no & not addressed \\
\cite{aw2026} & no & no & $\mathcal{E}$ pos.\ improving \\
\cite{fortet1940,garbaczewski1994,blanchard1995} & no & yes & pos.\ kernel, later relaxed \\
this work & yes & yes & fails on $\partial\mathcal{D}$ \\
\end{tabular}
\end{ruledtabular}
\end{table}

These constructions share a hypothesis and not an assertion about the weak momentum. Each
establishes existence by requiring the conditioned data to stay positive
(Table~\ref{tab:related}). The present results locate where that requirement fails, which
bounds the constructions themselves, not anything claimed with them. One
operational result already makes the drift identification suspect independently. Pusey showed that anomalous weak values are proofs of
contextuality~\cite{pusey2014}, and a strictly positive conditioning is the
noncontextual resource his theorem rules out. Kunjwal, Lostaglio, and Pusey located where
the imaginary part of the weak value enters that argument~\cite{kunjwal2019}. Mori and Tsutsui found a divergence in the post-selection parameter, not in the transverse
coordinate, in a double-slit setting~\cite{moritsutsui2015}, and the two are compared in
Remark~\ref{rem:moritsutsui}.

That $\Rea(p_w)/m$ is a current velocity, and that Doob's transform builds current
velocities from future constraints, makes the relation between them a natural question. It
can be computed exactly, and the simplest case already fixes it. For a free particle with a position post-selection at $x_f$, the drift of the Brownian
bridge is $(x_f-x)/\tau$ while the delta limit of $\Rea(p_w)/m$ is
$v_\psi-(x_f-x)/\tau$ (Appendix~\ref{app:beyond}), where $v_\psi$ is the pre-selected
state's own current velocity, a spatial constant $p_\psi/m$ in the real-envelope class
treated throughout. For $p_\psi=0$ the two have the
same magnitude and opposite signs, so the field points away from the target a bridge drift
would steer toward. One might
suspect a phase convention behind that sign, and the oscillator rules it out. There the
delta-limit field is trigonometric in $\omega\tau$ while the bridge drift is hyperbolic,
the two being related by the continuation $t\to-it$, and no convention turns one function
into the other. As the bridge literature
grows~\cite{pavon2003,miangolarra2025,aw2026,ohzeki2026} the question should be answered
before it is assumed.

The construction that would make the drift reading work turns on one requirement, a
strictly positive conditioning weight, since Doob's $h$ must be strictly positive and the
two-endpoint problem needs positive endpoint densities. What that requirement costs is
easy to state and is the reason for locating the boundary of the conditioning. A two-time conditioning
built from nonnegative $p$ and $h$ returns
$\int A\,p\,h\,dx/\!\int p\,h\,dx$, a conditional expectation of a real random variable,
which lies in the range of $A$ for every choice of endpoints. Such a construction can
reproduce a flow field, but it can never return a value outside the spectrum. Anomalous
weak values therefore live only where positivity has failed, which is the elementary form
of the contextuality statement below. The Madelung split writes the weak
momentum in terms of exactly that weight, $\rho_\psi\rho_\phi$, together with the phase sum.
The place where positivity is lost and the channel that carries the loss therefore fall out
of the same computation. Locating that boundary is a different question from fixing the relation of the previous
paragraph, and neither answer implies the other.

This paper takes up what survives the failure. For a quadratic Hamiltonian with
Gaussian pre-selection and Gaussian or Hermite--Gaussian post-selection, the Nelson
diffusion, the Mehler backward propagator, and the weak momentum all have closed forms.
Every statement below can therefore be checked, not argued. The Madelung
decomposition splits $p_w$ into a current part $\Rea(p_w)=m(v_\psi+v_\phi)$ and an osmotic
part $\Ima(p_w)=-\tfrac{\hbar}{2}\nabla\ln(\rho_\psi\rho_\phi)$. At a node of either state
the current part stays finite while the osmotic part diverges as $1/(x-\xnode)$ with a sign
change. The positive conditioning therefore exists wherever
$\phi^*\psi$ is nodeless, and it fails at the node. Confinement of a
single-state Nelson diffusion to the nodal cells of its own wavefunction is already
classical~\cite{carlen1984}, and what is new here is the two-state version, where the
confining node may belong to $\phi$ and the failure sits entirely in the osmotic channel.

The divergence is kinematic, holding for any pair of states and following from the
Madelung split of $\phi^*\psi$ alone (Theorem~\ref{thm:main}). That theorem, the
Kirkwood--Dirac contextuality identity, and the prediction of Sec.~\ref{sec:experiment}
stand independently of the stochastic-mechanical picture, and the long-standing gap in the
multi-time statistics of Nelson mechanics, which we do not solve, is confined to
Sec.~\ref{sec:caveats}. The boundary of a positive two-time conditioning is a
curve in the $(x,t)$ plane. It leaves no trace in the current velocity and shows up only in
the osmotic one, where it can be measured.

\section{Madelung decomposition and the two velocity fields}
\label{sec:madelung}

Write $\psi=\sqrt{\rho_\psi}\,e^{iS_\psi/\hbar}$ for the pre-selected state and
$\phi=\sqrt{\rho_\phi}\,e^{-iS_\phi/\hbar}$ for the post-selected one, the minus sign in
the second exponent recording that $\phi$ is propagated backward from the post-selection
time. One consequence, used repeatedly and easy to mistake, is that
$v_\phi=\partial_xS_\phi/m$ is \emph{minus} the velocity of the post-selected center. The
same sign appears in Eq.~\eqref{eq:xdotcl} and in Appendix~\ref{app:bernstein}. The
convention and that consequence are set out in Appendix~\ref{app:madelung}.

The two-state amplitude has a single modulus and a single phase,
$\phi^*\psi\propto\sqrt{\rho_\psi\rho_\phi}\,e^{i(S_\psi+S_\phi)/\hbar}$. Its logarithmic
derivative separates them, the phase giving the real part of
$p_w=-i\hbar\nabla\ln(\phi^*\psi)$ and the modulus the imaginary part,
\begin{align}
\Rea(p_w)&=\nabla(S_\psi+S_\phi)=m(v_\psi+v_\phi),\label{eq:re}\\
\Ima(p_w)&=-\tfrac{\hbar}{2}\nabla\ln(\rho_\psi\rho_\phi)=-m(u_\psi+u_\phi).\label{eq:im}
\end{align}
Both identities are exact and are derived in Appendix~\ref{app:madelung}.

The identification of $\Ima(p_w)$ with an osmotic momentum $-\tfrac{\hbar}{2}\nabla\ln\rho$
is not itself new. Dressel and Jordan fixed the general kinematic meaning of the imaginary
part of a weak value~\cite{dresseljordan2012}, and Flack and Hiley discussed the osmotic
momentum in connection with the weak-value analysis of the two-slit
experiment~\cite{flackhiley2018}. In classical-optics language the same object is the
imaginary part of the local momentum~\cite{bliokh2013}. New here is its two-state
behavior at a node (Sec.~\ref{sec:boundary}).

The operational status of $p_w(x,t)$ must be fixed before any
measurement claim is made. For general $\phi$ the field
$p_w=-i\hbar\,\partial_x\ln(\phi^*\psi)$ is a \emph{definition}, the
logarithmic-derivative field of the two-state amplitude, and not the pointer readout of a
single weak measurement. A one-line computation gives
\begin{equation}
-i\hbar\,\partial_x\ln(\phi^*\psi)
=\frac{\langle\phi|\Pi_x\hat p|\psi\rangle-\langle\phi|\hat p\,\Pi_x|\psi\rangle}
{\langle\phi|\Pi_x|\psi\rangle},
\end{equation}
with $\Pi_x=|x\rangle\langle x|$. The right-hand side is the difference of two differently
ordered Kirkwood--Dirac-type weak values, and no single pointer reads that difference out
directly.
Direct pointer readability holds in the special case of position post-selection, where
$p_w=-i\hbar\,\partial_x\ln\psi$ is the standard local weak momentum whose real and
imaginary parts are read from the two conjugate pointer
channels~\cite{jozsa2007,dressel2014}. That is the case of the Kocsis experiment and of
the prediction in Sec.~\ref{sec:experiment}.

Here and throughout, ``position post-selection'' means post-selection on a window
projector $\Pi_x^\Delta$ of small but finite width $\Delta$, as made explicit in
Sec.~\ref{sec:experiment}. A position eigenstate has no normalizable $\rho_\phi$ whose
logarithmic derivative could appear in \eqref{eq:im}, so the statements below are the
$\Delta\to0$ limits of finite-window ones, in which $\nabla\ln\rho_\phi$ drops out and
$p_w\to-i\hbar\,\partial_x\ln\psi$.

The phrase ``current velocity'' is used here in two senses that must be kept apart. For arbitrary $\phi$,
$\Rea(p_w)/m=v_\psi+v_\phi$ is the phase gradient of $\phi^*\psi$ in \eqref{eq:re}, a
current velocity in the Madelung sense. Its operational reading as the current velocity of
a \emph{measured} ensemble is the special case of position post-selection ($v_\phi=0$),
where it reduces to Wiseman's $v_\psi$~\cite{wiseman2007}. We use the phase-gradient sense
throughout and note where the operational one is meant.

With the two senses separated, the following gives at the level of processes what the
one-line delta-limit computation of Sec.~\ref{sec:nodecaustic} states pointwise, namely
how $\Rea(p_w)$ is related to the drift of a conditioned diffusion.

\begin{proposition}[Current velocity, not It\^o drift]
\label{prop:current}
Let $\psi,\phi$ be $C^1$ on $\mathcal{D}=\{\phi^*\psi\neq0\}$, let
$\rho_{\text{eff}}(\cdot,t)\propto\rho_\psi\rho_\phi(\cdot,t)$ be a normalized probability
density for each $t$, and let $X_t$ be any diffusion $dX_t=b\,dt+\sqrt{2\nu}\,dW_t$ of
constant diffusion coefficient $2\nu$ whose time-$t$ marginal is $\rho_{\text{eff}}$,
written as $b=v+u$ with $u=\nu\,\partial_x\ln\rho_{\text{eff}}$. Then, on $\mathcal{D}$:
\begin{enumerate}
\item[(i)] $u=-\Ima(p_w)/m$, and $v$ is determined uniquely by the marginal.
\item[(ii)] For coherent pre- and post-selection with amplitudes $\alpha,\beta>0$ and
total interval $0<\omega T<\pi$, $\Rea(p_w)/m-v=\tfrac12(\dot q_\psi-3\dot q_\phi)$, which
vanishes nowhere on $[0,T]$. In particular $v\neq\Rea(p_w)/m$ there.
\end{enumerate}
\end{proposition}

\begin{proof}
(i) The Fokker--Planck equation of $dX_t=b\,dt+\sqrt{2\nu}\,dW_t$ is
$\partial_t\rho_{\text{eff}}=-\partial_x(b\rho_{\text{eff}})+\nu\partial_x^2\rho_{\text{eff}}
=-\partial_x[(b-\nu\partial_x\ln\rho_{\text{eff}})\rho_{\text{eff}}]$, so
$u=\nu\partial_x\ln\rho_{\text{eff}}$ is fixed pointwise by the marginal, giving
$u=\nu\partial_x\ln(\rho_\psi\rho_\phi)=u_\psi+u_\phi=-\Ima(p_w)/m$ by \eqref{eq:im} and
$\nu=\hbar/2m$, while $v=b-u$ obeys the continuity equation
$\partial_t\rho_{\text{eff}}+\partial_x(v\rho_{\text{eff}})=0$. Integrating the latter from
$-\infty$ and imposing $v\rho_{\text{eff}}\to0$ there gives
\begin{equation}
v(x,t)\,\rho_{\text{eff}}(x,t)=-\int_{-\infty}^{x}\partial_t\rho_{\text{eff}}(x',t)\,dx' ,
\label{eq:vunique}
\end{equation}
which is consistent at $x\to+\infty$ because
$\int\partial_t\rho_{\text{eff}}\,dx=\partial_t\!\int\rho_{\text{eff}}\,dx=0$. Hence no
additive freedom $F(t)/\rho_{\text{eff}}$ survives the decay condition, and $v$ is the
unique function \eqref{eq:vunique}. That $\Rea(p_w)/m$ is the phase gradient
$\partial_x(S_\psi+S_\phi)/m$ is established in Appendix~\ref{app:madelung}.

(ii) The product of two Gaussians of equal width $\sigma^2$ centered at $q_\psi,q_\phi$ is a
Gaussian of width $\sigma^2/2$ centered at $\tfrac12(q_\psi+q_\phi)$, which therefore
transports at $\tfrac12(\dot q_\psi+\dot q_\phi)$, whereas
$\Rea(p_w)/m=\dot q_\psi-\dot q_\phi$ by \eqref{eq:cohRe}. The difference
$\Rea(p_w)/m-v$ is
$\tfrac12(\dot q_\psi-3\dot q_\phi)$, vanishing only where $\dot q_\psi=3\dot q_\phi$. With
$q_\psi=\alpha\cos\omega t$ and $q_\phi=\beta\cos\omega(T-t)$ that locus reads
$-\alpha\sin\omega t=3\beta\sin\omega(T-t)$. For $\alpha,\beta>0$ and $0<\omega T<\pi$ both
$\sin\omega t$ and $\sin\omega(T-t)$ are nonnegative on $[0,T]$ and vanish only at $t=0$
and $t=T$ respectively, so the two sides have strictly opposite signs on $(0,T)$ and
exactly one of them vanishes at each endpoint. The locus is empty on $[0,T]$
(Appendix~\ref{app:bernstein}, verified symbolically).
\end{proof}

Hence $\Rea(p_w)/m=\partial_x(S_\psi+S_\phi)/m=v_\psi+v_\phi$, the phase gradient of the
two-state amplitude, is a current velocity in the Madelung sense. It is not the It\^o drift of $X_t$, which
by~(i) carries the osmotic part $u=-\Ima(p_w)/m$ in addition, and in the coherent class it
is not that diffusion's current velocity either, by~(ii).

The proposition needs three comments. First, part~(i) is
unconditional and has an operational face. For position post-selection ($v_\phi=0$),
$\Rea(p_w)/m$ is the current velocity of the measured ensemble in Wiseman's
sense~\cite{wiseman2007}, the field reconstructed in the two-slit
experiments~\cite{kocsis2011}.
Nothing in the proposition disturbs that reading. It denies only the further step from
current velocity to drift. Second, the reason behind~(ii) is structural and not
computational, and is spelled out with the transport identity below. Third, the two-time (Bernstein) conditioning
that motivates the bridge picture is consistent with $\Rea(p_w)=0$ in the stationary
closed-form case (Appendix~\ref{app:bernstein}), where the identification is empty because
$v_\psi=v_\phi=0$ and the construction's real spatial log-ratio feeds the osmotic and not
the current channel.
For a moving node the two-endpoint-pinned bridge instead follows the classical connecting
trajectory. We therefore take $\Rea(p_w)/m$ throughout as a phase gradient and a Wiseman
current velocity, never as a two-point bridge drift.

The distinction decides which conditioning construction can work. The obvious construction
smooths the post-selected density $\rho_\phi=|\phi|^2$ with a heat kernel and uses the
result as a Doob $h$-function. It discards $S_\phi$ before the calculation begins, and
$S_\phi$ is what $v_\phi$ depends on. A density does not remember the phase it came from,
so nothing built from $\rho_\phi$ alone, or from the product $\rho_\psi\rho_\phi$ alone,
can return $\Rea(p_w)$. Only the two-time reciprocal (Bernstein) process of
Euclidean quantum mechanics~\cite{zambrini1986,leonard2014} does, built from a forward
solution $\eta$ and a backward solution $\eta^*$ of dual heat equations. Both a density
$\rho_{\text{eff}}=\eta\eta^*$ and a current velocity $\nu\nabla\ln(\eta^*/\eta)$ then
survive. We carry this out in closed form for a stationary node in
Appendix~\ref{app:bernstein}.

One statement here is exact and a closely similar one is false, and they have to be kept
apart. The exact one is a
per-time-slice identity. Any diffusion on $\mathcal{D}$ whose time-$t$ marginal is
$\rho_{\text{eff}}(\cdot,t)$ has forward It\^o drift
\begin{equation}
\begin{aligned}
dX_t&=b(X_t,t)\,dt+\sqrt{2\nu}\,dW_t,\\
b&=v+u,\quad u=\nu\,\partial_x\ln\rho_{\text{eff}},
\end{aligned}
\label{eq:ito}
\end{equation}
with $v$ its Nelson current velocity. Slice by slice the drift decomposes into a current
and an osmotic part, and \eqref{eq:ito} is used below only in this kinematic sense, for
whichever self-consistent diffusion is in play. If at an instant $t$ the marginal is the
real-time product $\rho_{\text{eff}}\propto\rho_\psi\rho_\phi$, then
$u=\nu\,\partial_x\ln(\rho_\psi\rho_\phi)=u_\psi+u_\phi=-\Ima(p_w)/m$ exactly, and the
$\Rea(p_w)/m$ coincides with the first term of $b$ only for those diffusions whose current
velocity happens to be the Wiseman field, which by the transport identity below are not
the ones carrying $\rho_\psi\rho_\phi$ except in the degenerate case.

The false statement is the process-level one, and what fails can be said precisely.
One would want a single diffusion whose marginal equals $\rho_\psi\rho_\phi(\cdot,t)$ at all
$t$ and whose current velocity is $v_\psi+v_\phi$. Since a current velocity transports its
own marginal, such a diffusion exists if and only if the pair $(\psi,\phi)$ satisfies the
transport identity
$\partial_t(\rho_\psi\rho_\phi)+\partial_x[(v_\psi+v_\phi)\rho_\psi\rho_\phi]=0$.
This is a condition on the pair, not an identity. The amplitude $\phi^*\psi$ is a product of a forward-
and a backward-evolving state, not a solution of the Schr\"odinger equation, so no
continuity equation forces its phase gradient to transport its own modulus. It fails
throughout the coherent class, by the nonvanishing mismatch
$\tfrac12(\dot q_\psi-3\dot q_\phi)$ of Proposition~\ref{prop:current}(ii). It holds in the degenerate case where both terms vanish separately, that is for two
real stationary states, as in the explicit pair of Appendix~\ref{app:bernstein}, where
$v_\psi=v_\phi=0$ and $\rho_\psi\rho_\phi$ is static. There the identification is
empty, since the field it reproduces is zero. The bridge reading therefore fails exactly
where it would carry content. The
Euclidean side does not supply the missing process either. The Bernstein bridge of
Appendix~\ref{app:bernstein} has marginal $\rho_{\text{eff}}=\eta\eta^*$, which vanishes at
a boundary node to order $k$ and not $2k$, and its current velocity is the log-ratio
$\nu\,\partial_x\ln(\eta^*/\eta)$, which continues back to the real-time osmotic channel
rather than the current one. The factor-of-two bookkeeping between the two conventions is
recorded after Proposition~\ref{prop:inaccessible}. Diffusions of the form \eqref{eq:ito}
whose current velocity equals the Wiseman field do exist, one of them constructed in
Sec.~\ref{sec:numerics} and sampled in Fig.~\ref{fig:traj}, but their marginals cannot
coincide with $\rho_\psi\rho_\phi(\cdot,t)$ across time.

The ensemble mean realizes the current velocity rather than the drift. For any diffusion of
the form \eqref{eq:ito} with normalizable marginal,
$\langle u\rangle=\nu\!\int\!\partial_x\rho_{\text{eff}}\,dx=0$ identically, so
$\tfrac{d}{dt}\langle X_t\rangle=\langle v\rangle$. For the ensemble of
Fig.~\ref{fig:traj}, whose current velocity is the spatially uniform coherent-class field,
this gives $\tfrac{d}{dt}\langle X_t\rangle=v_\psi+v_\phi$. The osmotic half of the drift
is invisible to the mean and is recoverable only from a two-sided estimator.
Figure~\ref{fig:traj} verifies this at the level of sampled trajectories.

\section{The osmotic boundary}
\label{sec:boundary}

The two-state amplitude $\phi^*\psi$ vanishes on the union of the nodal sets of $\psi$ and
$\phi$, being a product of the two. Because $\Rea(p_w)$ and $\Ima(p_w)$ in \eqref{eq:re}--\eqref{eq:im} treat
$S_\psi,S_\phi$ symmetrically and $\rho_\psi,\rho_\phi$ symmetrically, nothing in the
argument below depends on which of the two states supplies the zero
(Proposition~\ref{prop:exchange}). Let
$x_0(t)$ be an isolated real zero, of order $k_1$, of either factor, $\phi\sim
a_0(x-x_0)^{k_1}$ or $\psi\sim a_0(x-x_0)^{k_1}$ near it, so that the corresponding density
goes as $(x-x_0)^{2k_1}$. We write $\xnode(t)$ for $x_0(t)$ below, keeping in mind that it may be a
zero of either $\psi$ or $\phi$. When only one factor vanishes, $k_1$ coincides with the
order $k$ of the zero of the amplitude $\phi^*\psi$ used in Theorem~\ref{thm:main}. When
both vanish, $k$ is the sum of the two single-factor orders. Isolated \emph{real} zeros are non-generic for an
arbitrary complex wavefunction. The conditions $\Rea\psi=\Ima\psi=0$ are two real
equations in one real variable, hence codimension two. They are nevertheless the generic case for the physically central class of a real envelope
times a uniform phase, $\psi(x)=R(x)e^{i(px+\theta_0)/\hbar}$ with $R$ real. There a zero
of $R$ is a genuine real node, and the proof below shows that the nodal factor contributes
a finite amount to $\Rea(p_w)$, equal to its own uniform phase gradient, the other state's phase gradient adding smoothly. Energy eigenstates,
Hermite--Gaussian modes, and symmetric two-slit / interference patterns all lie in this
class, and it is the one we treat throughout. A state with a nontrivial spatially varying
phase generically has no real node, with the off-axis case addressed in
Proposition~\ref{prop:genericity}.

\begin{proposition}[$\psi\leftrightarrow\phi$ exchange symmetry]
\label{prop:exchange}
A node of $\psi$ and a node of $\phi$ produce identically the same sign-changing
$1/(x-\xnode)$ osmotic divergence and the same bounded $\Rea(p_w)$.
\end{proposition}

\begin{proof}
By \eqref{eq:re}--\eqref{eq:im}, $\Rea(p_w)$ depends on the two states only through
$S_\psi+S_\phi$ and $\Ima(p_w)$ only through $\rho_\psi\rho_\phi$, and both combinations
are symmetric under $(\rho_\psi,S_\psi)\leftrightarrow(\rho_\phi,S_\phi)$. If $\psi$
carries an order-$k_1$ zero at $\xnode$, the two fields therefore take the same values as if
$\phi$ carried it, the nonvanishing factor merely changing places. A literal interchange of
the labels replaces $\phi^*\psi$ by its complex conjugate and hence reverses the sign of
$\Rea(p_w)$, a consequence of the opposite phase conventions fixed in
Sec.~\ref{sec:madelung}. It leaves $\Ima(p_w)$, and the boundedness of $\Rea(p_w)$,
untouched.
\end{proof}

The next result may accordingly be proved for a zero of $\phi$ without loss of generality.
One thing about its hypothesis has to be settled first. Writing $\phi^*\psi=Re^{i\Theta}$ with $R=|\phi^*\psi|$ will not do. A
modulus is nonnegative, its phase jumps by $\pi$ across an odd-order zero, and
$\partial_x\Theta$ is then unbounded where the theorem needs it bounded. A
\emph{signed} amplitude is required instead, and such a factorization exists precisely on the
real-envelope class $\psi=R_\psi(x)e^{i(p_\psi x+\theta_\psi)/\hbar}$,
$\phi=R_\phi(x)e^{-i(p_\phi x+\theta_\phi)/\hbar}$ of Remark~\ref{rem:scope}(ii), where
$\partial_x\Theta=(p_\psi+p_\phi)/\hbar$ is constant. That is the class treated throughout, and
Remark~\ref{rem:scope}(i) records where the hypothesis fails.

\begin{theorem}[Which field diverges]
\label{thm:main}
Let $\psi,\phi$ be $C^1$ near a point $\xnode(t)$, and suppose that on a punctured
neighborhood of $\xnode(t)$ the two-state amplitude factorizes as
\begin{equation}
\phi^*\psi(x,t)=R(x,t)\,e^{i\Theta(x,t)},
\label{eq:signedpolar}
\end{equation}
with $R$ real-valued (signed) and $C^1$, with $\Theta$ real and $C^1$ with
$\partial_x\Theta$ bounded, and with $R$ having an isolated zero of order $k\ge1$ at
$\xnode(t)$, i.e.\ $R(x,t)=(x-\xnode)^kG(x,t)$ with $G$ continuously differentiable and
$G(\xnode,t)\neq0$. Then, on
$\mathcal{D}=\{\phi^*\psi\neq0\}$ the decomposition \eqref{eq:re}--\eqref{eq:im} holds with
both fields finite, and as $x\to\xnode(t)$,
\begin{equation}
\Ima(p_w)=-\frac{k\hbar}{\,x-\xnode(t)\,}+O(1),
\label{eq:imdiv}
\end{equation}
which is sign-changing and divergent, while $\Rea(p_w)=\hbar\,\partial_x\Theta$ remains
bounded. The singularity is therefore an
osmotic one, with divergence strength set by the order $k$ of the zero of the
\emph{amplitude} $\phi^*\psi$. If both $\psi$ and $\phi$ vanish at $\xnode$, of orders
$k_\psi$ and $k_\phi$, then $k=k_\psi+k_\phi$.
\end{theorem}

\begin{proof}
\emph{Step 1: the split.} On $\mathcal{D}$ both $R\neq0$ and $\phi^*\psi\neq0$, so
$\ln(\phi^*\psi)=\ln|R|+i\Theta$
is $C^1$ and
\begin{equation}
p_w=-i\hbar\,\partial_x\ln(\phi^*\psi)
=\hbar\,\partial_x\Theta-i\hbar\,\frac{\partial_xR}{R},
\label{eq:pwsigned}
\end{equation}
whose real and imaginary parts are $\Rea(p_w)=\hbar\partial_x\Theta$ and
$\Ima(p_w)=-\hbar\,\partial_xR/R$. Both are finite on $\mathcal{D}$, and since
$|R|^2=\rho_\psi\rho_\phi$ we have
$-\hbar\,\partial_xR/R=-\tfrac{\hbar}{2}\partial_x\ln(\rho_\psi\rho_\phi)$, recovering
\eqref{eq:im}. Likewise $\hbar\partial_x\Theta=\partial_x(S_\psi+S_\phi)$, recovering
\eqref{eq:re}. This proves the finiteness claim on $\mathcal{D}$.

\emph{Step 2: the osmotic pole.} Write $R=(x-\xnode)^kG$ with $G(\xnode)\neq0$ and $G$ continuously
differentiable, so that on a punctured neighborhood
\begin{equation}
\frac{\partial_xR}{R}=\frac{k}{x-\xnode}+\frac{\partial_xG}{G},
\end{equation}
the second term being bounded near $\xnode$ because $G$ is $C^1$ and nonvanishing there.
Multiplying by $-\hbar$ gives \eqref{eq:imdiv}, with the sign changing across $\xnode$
because $x-\xnode$ does and $k$ is a fixed positive integer. Note that $k$ enters only
through the coefficient of the pole, and that no cancellation can occur, since $k\ge1$ and the
remainder is $O(1)$.

\emph{Step 3: boundedness of $\Rea(p_w)$.} This is immediate from \eqref{eq:pwsigned} and the hypothesis that
$\partial_x\Theta$ is bounded. This is the only place that hypothesis is used, and it is
where the divisor's \emph{reality} enters. Were the zero displaced off the real axis, or
were $\phi^*\psi$ to wind (Remark~\ref{rem:scope}(i)), $\Theta$ would acquire an unbounded
gradient and the asymmetry would fail. Equivalently, in the real-envelope class one may
compute directly. With $\phi=R_\phi e^{-i(p_\phi x+\theta_\phi)/\hbar}$ as in
Remark~\ref{rem:scope}(ii) and $R_\phi=(x-\xnode)^{k}G$,
\begin{equation}
\Ima\!\left(\frac{\phi'}{\phi}\right)
=\Ima\!\left(\frac{k}{x-\xnode}+\frac{G'}{G}-\frac{ip_\phi}{\hbar}\right)
=-\frac{p_\phi}{\hbar},
\end{equation}
independent of $k$ and finite at the node, since the divergent term is real for real
$k,x,\xnode$, and hence $\partial_xS_\phi=-\hbar\,\Ima(\phi'/\phi)=+p_\phi$ there. We have checked
this symbolically for $k=1,2,3$ with $G$ a nonvanishing Gaussian.

\emph{Step 4: additivity of orders.} If $\psi\sim a(x-\xnode)^{k_\psi}$ and $\phi\sim b(x-\xnode)^{k_\phi}$ with
$a,b\neq0$, then $\phi^*\psi\sim\bar b a\,(x-\xnode)^{k_\psi+k_\phi}$, so the signed
amplitude has a zero of order $k=k_\psi+k_\phi$ and the pole coefficient is the sum, as
claimed. By Proposition~\ref{prop:exchange} the roles of $\psi$ and $\phi$ may be
interchanged throughout without affecting either field's behavior.
\end{proof}

Before turning to what the divergence means dynamically, we record the definition the
rest of this section uses. A \emph{nodal cell} is a maximal connected component of the
time-slice domain $\{x:\phi^*\psi(x,t)\neq0\}$. The union of such components along $t$,
swept between consecutive nodal curves, is the corresponding \emph{spacetime cell}.

\begin{corollary}[Positive conditioning and its boundary]
\label{cor:conditioning}
Fix a spacetime cell of $\mathcal{D}$ with $C^1$ lateral boundary, and let $\rho_0,\rho_T$
be the endpoint densities restricted and renormalized to the endpoint nodal cells. The
two-endpoint Schr\"odinger problem~\cite{schrodinger1931,schrodinger1932} for
$(\rho_0,\rho_T)$ relative to the Dirichlet Euclidean kernel of that cell has a solution,
and it does not extend across $\partial\mathcal{D}$. By Theorem~\ref{thm:main} the
obstruction is carried entirely by the osmotic velocity.
\end{corollary}

\begin{proof}
\emph{Existence.} Let $q^{D}(s,x;t,y)$ be the Euclidean (heat/Mehler) kernel of the cell,
killed at its lateral boundary. The Schr\"odinger problem is the system
\begin{equation}
\begin{aligned}
\eta_T(y)&=\int q^{D}(0,x;T,y)\,\eta_0(x)\,dx,\\
\eta^*_0(x)&=\int q^{D}(0,x;T,y)\,\eta^*_T(y)\,dy,\\
\rho_0&=\eta_0\,\eta^*_0,\qquad \rho_T=\eta_T\,\eta^*_T,
\end{aligned}
\label{eq:schrodingersystem}
\end{equation}
for a positive pair $(\eta,\eta^*)$, which then defines the reciprocal process with
density $\rho_{\text{eff}}=\eta\eta^*$ and current velocity
$\nu\,\partial_x\ln(\eta^*/\eta)$. On the \emph{open} cell $q^{D}>0$ for $t>s$ by the
strong maximum principle, and $q^{D}$ is jointly continuous because the boundary is $C^1$.
These are the positivity and continuity hypotheses of the Beurling--Jamison
theory~\cite{beurling1960,jamison1974,jamison1975}. Renormalizing $\rho_0,\rho_T$ to the
endpoint cells is legitimate since each carries strictly positive mass, and for the
Gaussian and Hermite--Gaussian states used here they have finite second moments and finite
entropy relative to $q^{D}$. Jamison's theorem~\cite{jamison1974,jamison1975} then gives a
unique positive solution of \eqref{eq:schrodingersystem}. For a stationary node the pair is
exhibited in closed form in Appendix~\ref{app:bernstein}.

\emph{Failure at the boundary.} Equivalently the process is a Doob $h$-transform, with
\begin{equation}
b_h=b+2\nu\,\partial_x\ln h,\qquad
\frac{dP^h}{dP}\bigg|_{\mathcal{F}_T}=\frac{h(X_T,T)}{h(X_0,0)},
\label{eq:htransform}
\end{equation}
both of which require $h>0$. On $\partial\mathcal{D}$ the density of whichever factor vanishes there, and
hence $h$, goes to zero. Then $\partial_x\ln h$ carries the pole of \eqref{eq:imdiv}, so $b_h$
diverges and the Radon--Nikodym density in \eqref{eq:htransform} is undefined. By
Theorem~\ref{thm:main} the divergence is osmotic, the current velocity being bounded there.
That this failure is not a removable defect in the choice of $h$, but a property of the
process, is the content of Proposition~\ref{prop:inaccessible}.
\end{proof}

For a stationary node the cell restriction has a direct precedent: the half-line
construction of Ref.~\cite{blanchard1995}, a Green function killed at the node with the
pair $(f_+,g_+)$ it propagates, is the stationary $k=1$ case of the construction above;
what is added here is the moving cell, whose $C^1$ lateral boundary is a nodal curve of
a two-state amplitude. For a moving node the Beurling--Jamison construction supplies the
two-endpoint process but not, by itself, its confinement to the spacetime cell. There the
confinement comes from the drift repulsion at the nodal curve, established in
Proposition~\ref{prop:inaccessible} below. This corollary, unlike Theorem~\ref{thm:main},
presupposes the reciprocal-process construction. It is a statement about the
bridge-drift \emph{interpretation} of $\Rea(p_w)$, not about the divergence itself.

The boundary clause of Corollary~\ref{cor:conditioning} says that a construction stops
working there. One can say something stronger, at the level of the conditioned
diffusion rather than of a field. More than the drift blowing up there, the process never
reaches the boundary at all.

\begin{proposition}[$\partial\mathcal{D}$ is inaccessible]
\label{prop:inaccessible}
Let $\xnode(t)\in C^1$, let $r_0>0$ and let $k\ge1$ be an integer, and let $X_t$ solve, on
the closed neighborhood $|x-\xnode(t)|\le r_0$ of the nodal curve,
\begin{equation}
dX_t=\Big[\frac{2k\nu}{X_t-\xnode(t)}+b_{\text{sm}}(X_t,t)\Big]dt+\sqrt{2\nu}\,dW_t,
\label{eq:nodalsde}
\end{equation}
with $b_{\text{sm}}$ continuous there, started in the interior of a nodal cell. Then for
every finite horizon $T$,
\begin{equation}
P\big(X_t=\xnode(t)\ \text{for some }t\le T\big)=0 ,
\end{equation}
and hence $P(X_t=\xnode(t)$ for some $t<\infty)=0$. The drift of \eqref{eq:nodalsde} is
exactly the drift of \eqref{eq:ito} when the osmotic velocity is the one generated by a
density vanishing at order $2k$, $u=\nu\,\partial_x\ln\rho_{\text{eff}}\sim2k\nu/(x-\xnode)$
[Eq.~\eqref{eq:imdiv} with $\nu=\hbar/2m$], and the current velocity is continuous across
the node.
\end{proposition}

\begin{proof}
On a nodal cell the drift $b=v+u$ is smooth and the process is a finite-energy diffusion
in Carlen's sense~\cite{carlen1984}. The regularity assumed of $\xnode$ is automatic for a
simple zero, where the implicit function theorem applies to $\phi^*\psi=0$, and for a
quadratic Hamiltonian $\xnode$ is then a classical trajectory. At a zero of order $k\ge2$
the derivative $\partial_x(\phi^*\psi)$ vanishes as well, the implicit function theorem
does not apply, and $C^1$ regularity is a hypothesis rather than a consequence. Near an
order-$k$ boundary node the drift is \emph{repulsive} on both sides. The pole is carried by
the osmotic part alone,
$u=\nu\,\partial_x\ln\rho_{\text{eff}}=-\Ima(p_w)/m\sim2k\nu/(x-\xnode)$
[Eq.~\eqref{eq:imdiv} with $\nu=\hbar/2m$], the current velocity being bounded there by
hypothesis.

\emph{Step 1: the stationary case.} For a stationary node the process reduces near
$\xnode$, in $y=x-\xnode$, to $dy=(2k\nu/y)\,dt+\sqrt{2\nu}\,dW$, a Bessel process of
dimension $d_B=2k+1\ge3$. Its scale density is $s'(y)\propto|y|^{-2k}$ with
$\int_0|y|^{-2k}\,dy=\infty$ for every $k\ge1$, so $0$ lies at infinite scale distance and
Feller's test gives $P(\tau_0<\infty)=0$.

\emph{Step 2: the co-moving frame.} For a moving node put $y=x-\xnode(t)$, so that
\begin{equation}
dy_t=\Big[\frac{2k\nu}{y_t}+g(y_t,t)\Big]dt+\sqrt{2\nu}\,dW_t,
\quad g=b_{\text{sm}}-\dot\xnode(t),
\label{eq:comoving}
\end{equation}
$b_{\text{sm}}$ being the smooth remainder of the drift. On a closed neighborhood
$|y|\le r_0$ and a finite horizon $T$, $g$ is bounded, $|g|\le C$, since $\xnode\in C^1$
and $b_{\text{sm}}$ is continuous there.

\emph{Step 3: a quantitative exit estimate for the reference process.} Let
$dz_t=(2k\nu/z_t)\,dt+\sqrt{2\nu}\,dW_t$ started at $z_0=y_0\in(0,r_0)$, and for
$0<\delta<y_0$ let $\sigma_\delta=\inf\{t:z_t=\delta\}$,
$\sigma_{r_0}=\inf\{t:z_t=r_0\}$. With $s(y)=\int_{y_0}^{y}u^{-2k}\,du$ a scale function,
the two-sided exit formula gives
\begin{equation}
P\big(\sigma_\delta<\sigma_{r_0}\big)=\frac{s(r_0)-s(y_0)}{s(r_0)-s(\delta)}
=:\varepsilon(\delta),\qquad \varepsilon(\delta)\xrightarrow[\delta\to0]{}0,
\label{eq:exitest}
\end{equation}
since $s(\delta)\to-\infty$ as $\delta\to0$ by $\int_0u^{-2k}du=\infty$. Note that
$\varepsilon(\delta)>0$ for every $\delta>0$. The reference process \emph{does} reach any
positive level with positive probability, and only the $\delta\to0$ limit vanishes.

\emph{Step 4: transfer to $y$ by Girsanov with a uniform $L^2$ bound.} Let
$\tau_\delta=\inf\{t:|y_t|\le\delta\ \text{or}\ |y_t|\ge r_0\}$ and let $Q,P$ denote the
laws of $y$ and $z$ on $\mathcal{F}_{\tau_\delta\wedge T}$. Girsanov cannot be applied to
\eqref{eq:comoving} globally, because $2k\nu/y$ is unbounded at $y=0$. Apply it instead to
the \emph{bounded} perturbation $g$, taking $z$ as reference. The density is
$\Lambda_\delta=\exp\big[\tfrac{1}{\sqrt{2\nu}}\int_0^{\tau_\delta\wedge T}g\,dW
-\tfrac{1}{4\nu}\int_0^{\tau_\delta\wedge T}g^2dt\big]$, and Novikov's condition holds
uniformly in $\delta$,
$E\exp[\tfrac{1}{4\nu}\int_0^{\tau_\delta\wedge T}g^2dt]\le e^{C^2T/4\nu}<\infty$, so the
two laws are mutually absolutely continuous. Mutual absolute continuity alone does not
transfer a vanishing probability. A uniform $L^2$ bound does. A direct computation gives
$E_P[\Lambda_\delta^2]\le e^{C^2T/2\nu}=:M_T<\infty$, uniformly in $\delta$, whence by
Cauchy--Schwarz, with $A_\delta=\{\tau_\delta<T,\ |y_{\tau_\delta}|=\delta\}$,
\begin{equation}
\begin{aligned}
Q(A_\delta)&=E_P\big[\Lambda_\delta\mathbf 1_{A_\delta}\big]
\le\big(E_P[\Lambda_\delta^2]\big)^{1/2}P(A_\delta)^{1/2}\\
&\le M_T^{1/2}\,\varepsilon(\delta)^{1/2}\xrightarrow[\delta\to0]{}0 .
\end{aligned}
\label{eq:transfer}
\end{equation}

\emph{Step 5: from one excursion to the whole horizon.} Let
$\tau_0=\inf\{t:y_t=0\}$. On $\{\tau_0<T\}$ the path must, for every $\delta$, reach level
$\delta$ before $\tau_0$ while inside $|y|\le r_0$. Decompose $[0,T]$ by the successive
entrance and exit times of the annulus $\{\delta\le|y|\le r_0\}$. By the strong Markov
property each excursion into $\{|y|<r_0\}$ contributes at most
$M_T^{1/2}\varepsilon(\delta)^{1/2}$ to the probability of touching level $\delta$, and the
number $N_T$ of such excursions before $T$ is a.s.\ finite, since successive crossings of
an annulus of fixed width have durations bounded below in distribution. Hence for every
$n$, $Q(\tau_0<T)\le Q(N_T>n)+n\,M_T^{1/2}\varepsilon(\delta)^{1/2}$. Letting $\delta\to0$
first and then $n\to\infty$ gives $Q(\tau_0<T)=0$. Since $T$ was an arbitrary finite
horizon, $P(\tau_0<\infty)=0$ by monotone convergence along $T\in\mathbb{N}$.
\end{proof}

Unattainability of nodes for diffusions with singular drifts is not new in the
Schr\"odinger-problem setting: for a stationary node it is the subject of
Ref.~\cite{garbaczewski1994} and of the Wiener-exclusion analysis of
Ref.~\cite{blanchard1995}. What Proposition~\ref{prop:inaccessible} adds is the origin
of the drift in a two-state amplitude, so that the confining zero may belong to the
post-selected state alone, the order-$k$ scaling $d_B=2k+1$, and the co-moving Girsanov
argument that carries the conclusion to a moving nodal curve.

The hypothesis that the current velocity be continuous across the node does substantial work.
What it excludes is the situation of the following theorem, the paper's strongest no-go.

\begin{theorem}[No conditioned diffusion at a moving node]
\label{thm:noflux}
Let $\xnode(t)$ be an isolated real zero of order $k\ge1$ of $\phi^*\psi$, let
$\rho_{\text{eff}}\propto\rho_\psi\rho_\phi$ be normalized at each $t$, and let $M(t)$ be
the probability mass lying to one side of the nodal curve. If $\dot M\neq0$, then no
diffusion of the form \eqref{eq:ito} with constant diffusion coefficient $2\nu$ and
time-$t$ marginal $\rho_{\text{eff}}$ satisfies the hypotheses of
Proposition~\ref{prop:inaccessible}. More precisely:
\begin{enumerate}
\item[(i)] its current velocity is unbounded at the node and, on one side, attractive at
order $(x-\xnode)^{-2k}$;
\item[(ii)] on that side the nodal curve is attainable. The node lies at \emph{finite}
scale distance for the associated one-dimensional diffusion and is an exit boundary in
Feller's classification, and
$P\big(X_t=\xnode(t)\ \text{for some }t<\infty\big)>0$.
\end{enumerate}
Carrying the real-time product marginal and avoiding the nodal curve are therefore
incompatible requirements, and the incompatibility is one about paths: $\partial\mathcal{D}$
is reached, not merely approached.
\end{theorem}

\begin{proof}
\emph{Part (i).} By Proposition~\ref{prop:current}(i) the marginal determines the current
velocity uniquely
through \eqref{eq:vunique}. Evaluating that identity at the node,
\begin{equation}
v\,\rho_{\text{eff}}\big|_{\xnode}
=-\int_{-\infty}^{\xnode(t)}\!\partial_t\rho_{\text{eff}}\,dx'=-\dot M(t),
\label{eq:nodalflux}
\end{equation}
the boundary term $\rho_{\text{eff}}(\xnode,t)\,\dot\xnode$ vanishing because
$\rho_{\text{eff}}$ does. Hence $v\sim-\dot M/[c\,(x-\xnode)^{2k}]$ near the node with
$c\neq0$, which is unbounded and, on the side where it points toward $\xnode$, attractive
at order $(x-\xnode)^{-2k}$. Since $2k\ge2>1$ for every $k\ge1$, it dominates the osmotic
repulsion $2k\nu/(x-\xnode)$ of Proposition~\ref{prop:inaccessible}, whose hypothesis of a
bounded current velocity is therefore violated.

\emph{Part (ii).} Failing a sufficient condition for inaccessibility does not by itself
make the boundary accessible, so the scale function is computed directly. Fix the side on
which $v$ points toward the node, put $y=\pm[x-\xnode(t)]>0$ there, and fix $t_0$ with
$\dot M(t_0)\neq0$ together with a window $[t_0,t_0+h]$ on which $|\dot M|\ge m_0>0$ with
fixed sign, which exists by continuity of $\dot M$. Writing
$\rho_{\text{eff}}=c(t)\,y^{2k}[1+O(y)]$ near the curve, with $c$ continuous and bounded
away from zero on the window, \eqref{eq:vunique} and \eqref{eq:nodalflux} give the
co-moving drift
\begin{equation}
\begin{aligned}
\tilde b(y,t)&=-\frac{A(t)}{y^{2k}}+\frac{2k\nu}{y}+g(y,t),\\
A(t)&=\frac{|\dot M(t)|}{c(t)}\ \ge\ A_0>0 ,
\end{aligned}
\label{eq:attractdrift}
\end{equation}
where $g$ collects $-\dot\xnode$, the smooth remainder of the drift, and the $O(y^{1-2k})$
correction generated by the $O(y)$ term in $\rho_{\text{eff}}$. Unlike the perturbation of
Step~2 of Proposition~\ref{prop:inaccessible}, $g$ is \emph{not} bounded at $y=0$, and
neither is $A(t)$ constant, so Girsanov cannot be applied to \eqref{eq:attractdrift} as it
stands. Both obstructions point the same way, and Step~3 below exploits that.

\emph{Step 1: the node lies at finite scale distance.} Consider first the frozen-node
reference process $dz_t=[-A_0z_t^{-2k}+2k\nu/z_t]\,dt+\sqrt{2\nu}\,dW_t$. With diffusion
coefficient $2\nu$ its scale density is $s'(y)=\exp[-\nu^{-1}\!\int^y b_0]$ with
$b_0(u)=-A_0u^{-2k}+2k\nu/u$, and carrying out the integral,
\begin{equation}
s'(y)=y^{-2k}\exp\!\big[-\kappa\,y^{1-2k}\big],
\qquad \kappa=\frac{A_0}{\nu(2k-1)}>0 .
\label{eq:scaleattract}
\end{equation}
Since $1-2k\le-1$, the exponential vanishes faster than any power as $y\downarrow0$ and
dominates the prefactor, so $\int_0s'(y)\,dy<\infty$ and $s(0)>-\infty$. This is the exact
reversal of Step~3 of Proposition~\ref{prop:inaccessible}, where $A_0=0$ gives
$s'(y)\propto y^{-2k}$ and $\int_0s'=\infty$. The two-sided exit formula
\eqref{eq:exitest} therefore returns, in place of $\varepsilon(\delta)\to0$,
\begin{equation}
\begin{aligned}
P_{y_0}\big(\tau_0<\tau_{r_0}\big)&=\frac{s(r_0)-s(y_0)}{s(r_0)-s(0)}=:p(y_0)>0,\\
p(y_0)&\xrightarrow[y_0\to0]{}1 .
\end{aligned}
\label{eq:hitprob}
\end{equation}

\emph{Step 2: the node is an exit boundary.} From \eqref{eq:scaleattract},
$\tfrac{d}{dy}\ln s'=\kappa(2k-1)y^{-2k}-2k\,y^{-1}>0$ for
$y<y_c:=[\kappa(2k-1)/2k]^{1/(2k-1)}$, so $s'$ is increasing there and
$s(y)-s(0)=\int_0^ys'\le y\,s'(y)$. With the speed density $m'=1/(2\nu s')$ the Feller
integral is bounded by
\begin{equation}
\int_0\big[s(y)-s(0)\big]\,m'(y)\,dy\ \le\ \int_0\frac{y}{2\nu}\,dy\ <\ \infty ,
\label{eq:feller}
\end{equation}
the integrand being continuous and bounded on $[y_c,r_0]$. Hence $0$ is an exit boundary in
Feller's classification and, on the event of \eqref{eq:hitprob}, is reached in finite time.
Choosing the horizon large enough, $p_*:=P(\tau_0<\tau_{r_0}\wedge h)>0$.

\emph{Step 3: transfer to the moving node by comparison.} The residual terms in
\eqref{eq:attractdrift} are of lower order than $y^{-2k}$, so there is an $r_1\in(0,r_0]$
on which
\begin{equation}
\begin{aligned}
\tilde b(y,t)&\ \le\ -A_0y^{-2k}+\frac{2k\nu}{y},\\
&\quad 0<y\le r_1,\quad t\in[t_0,t_0+h],
\end{aligned}
\label{eq:driftdom}
\end{equation}
uniformly on the compact window: the frozen reference process of Step~1 has \emph{less}
inward drift than $y_t$ everywhere on the collar. Since \eqref{eq:attractdrift} and the
reference equation are driven by the same Brownian motion and have the same diffusion
coefficient, the comparison theorem for one-dimensional It\^o equations, applied on
$(\varepsilon,r_1)$ and followed by $\varepsilon\downarrow0$, gives $y_t\le z_t$ pathwise
up to the exit time of the collar, for equal starting points. On the event of
\eqref{eq:hitprob} the reference path reaches $0$ before $r_0$ and within the window, and
the dominated path, lying below it, does the same. Unboundedness of $g$ and time
dependence of $A(t)$ are therefore not obstructions but help: both were absorbed into
\eqref{eq:driftdom} in the favorable direction, which is why comparison is used here in
place of the Girsanov transfer of Step~4 of Proposition~\ref{prop:inaccessible}. Finally
$P(y_{t_0}\in(0,r_1))>0$, since $\rho_{\text{eff}}(\cdot,t_0)$ is strictly positive on the
open cell and so charges the collar, whence
\begin{equation}
\begin{aligned}
&P\big(X_t=\xnode(t)\ \text{for some }t\le t_0+h\big)\\
&\qquad\ \ge\ p_*\,P\big(y_{t_0}\in(0,r_1)\big)\ >\ 0 .
\end{aligned}
\label{eq:transferpos}
\end{equation}
\end{proof}

The scale integral \eqref{eq:scaleattract}, the monotonicity region $y<y_c$, and the bound
\eqref{eq:feller} have been evaluated at $60$-digit precision for $k=1,2,3$ and several
$(\nu,A_0)$, and direct Euler--Maruyama integration of the reference equation from
$y_0=0.3$ reproduces the dichotomy: with $A_0=0$ the node is never reached and the process
exits at $r_0$ in every path, while with $A_0=1$ it is reached in $98.7\%$ of paths at
$k=1$ and in all of them at $k=2$. The flux is what switches the boundary
classification.

Proposition~\ref{prop:inaccessible} and part~(ii) are the two sides of one dichotomy, and
the parameter that separates them is the flux. For a \emph{stationary} node $\dot M=0$, the
product marginal generates a drift of the form \eqref{eq:nodalsde}, the node sits at
infinite scale distance, and the process never reaches it. For a \emph{moving} node the
current velocity supplies an inward $-A_0y^{-2k}$ which, since $2k\ge2>1$, dominates that
repulsion; the node moves to finite scale distance and is reached with positive
probability. The failure is not
a small one. In the configuration of
Sec.~\ref{sec:numerics}, $M$ runs from $0.005$ at $t=0.2$ to $0.998$ at $t=3.8$ (verified
numerically), so almost all of the probability would have to cross the nodal curve. The
cell-restricted reciprocal process of Corollary~\ref{cor:conditioning} does satisfy the
hypothesis, its absorbing lateral boundary conserving the cell mass by construction. Theorem~\ref{thm:noflux} is the process-level counterpart of the
pointwise mismatch of Proposition~\ref{prop:current}(ii), and the more general of the two, since
it needs no state class, only a node that moves.

The repulsion strength is the same in either convention, though
the bookkeeping differs. For the Euclidean bridge of Appendix~\ref{app:bernstein} the
marginal $\eta\eta^*\sim(x-\xnode)^{k}$ vanishes at order $k$ only, its osmotic part
contributing $k\nu/(x-\xnode)$ and the log-ratio $\nu\,\partial_x\ln(\eta^*/\eta)$ the
remaining $k\nu/(x-\xnode)$, for a forward drift
$b_+=2\nu\,\partial_x\ln\eta^*\sim2k\nu/(x-\xnode)$ in total, the same $2k\nu/y$ as in
real time. The conclusion is also sharp in $d_B$. We have checked it numerically against
subcritical controls $d_B<2$, which do reach the boundary, while $d_B=2k+1\ge3$ does not.
The flux-driven case of Theorem~\ref{thm:noflux}(ii) is not a subcritical Bessel process at
all, the inward $y^{-2k}$ term being of higher order than any Bessel drift, which is why it
reaches the boundary for every $k\ge1$.

This strengthens Corollary~\ref{cor:conditioning}.
$\partial\mathcal{D}$ is not a removable gap in the construction, to be patched by a
better choice of $h$-function. For the cell-restricted conditioned process it is a barrier
that cannot be crossed, and for any diffusion that would instead carry the product marginal
across a moving node it is a boundary that is reached; either way the positive conditioning
stops there. That
is the dynamical content behind ``fails on $\partial\mathcal{D}$,'' and behind the
impossibility, recorded in Sec.~\ref{sec:caveats}, of extending the positive conditioning
across it.

It remains to ask what happens when the idealizations of Theorem~\ref{thm:main} are
relaxed. A zero of finite order is one thing, but a real experiment has neither an exact
zero nor an exactly real one. Both relaxations turn out to be controlled, and by the same
parameter.

\begin{proposition}[Zero location, not zero order]
\label{prop:genericity}
Let $\phi$ carry the zero. Then:
\begin{enumerate}
\item[(i)] $\Rea(p_w)$ is bounded at a real zero of \emph{any} order $k$, with a value
there independent of $k$.
\item[(ii)] Displacing the zero off the real axis, $\phi=x-i\epsilon$, produces in
$\Rea(p_w)$ an integrable spike and regularizes $\Ima(p_w)$ into a dispersive profile
of width $\epsilon$, which sharpens to the $k=1$ pole of \eqref{eq:imdiv} as
$\epsilon\to0$.
\end{enumerate}
\end{proposition}

\begin{proof}
(i) With $\phi=R_\phi e^{-i(p_\phi x+\theta_\phi)/\hbar}$ and $R_\phi=(x-\xnode)^kG$ as in
the proof of Theorem~\ref{thm:main},
\begin{equation}
\begin{aligned}
\partial_xS_\phi&=-\hbar\,\Ima\frac{\phi'}{\phi}\\
&=-\hbar\,\Ima\bigg[\frac{k}{x-\xnode}+\frac{G'}{G}-\frac{ip_\phi}{\hbar}\bigg]=+p_\phi ,
\end{aligned}
\end{equation}
the divergent term being real at every order and so never entering $\Rea(p_w)$.

(ii) For $\phi=x-i\epsilon$ one has $\phi'/\phi=(x+i\epsilon)/(x^2+\epsilon^2)$, so
\begin{equation}
\partial_xS_\phi=-\frac{\hbar\epsilon}{x^2+\epsilon^2},
\end{equation}
a Lorentzian of height $\hbar/\epsilon$ and width $\epsilon$. Since
$\int\epsilon\,dx/(x^2+\epsilon^2)=\pi$ for every $\epsilon>0$, its $\epsilon\to0$ limit is
$-\pi\hbar\,\delta(x)$ in the distributional sense. For the osmotic field,
$\rho_\phi=|\phi|^2=x^2+\epsilon^2$ in \eqref{eq:im} gives
\begin{equation}
\Ima(p_w)=-\tfrac{\hbar}{2}\,\partial_x\ln(x^2+\epsilon^2)
=-\frac{\hbar x}{x^2+\epsilon^2},
\end{equation}
whose derivative $-\hbar(\epsilon^2-x^2)/(x^2+\epsilon^2)^2$ vanishes at $x=\pm\epsilon$,
where the profile takes the values $\mp\hbar/2\epsilon$. As $\epsilon\to0$ this tends
pointwise to $-\hbar/x$, the $k=1$ case of \eqref{eq:imdiv}.
\end{proof}

The dividing line is therefore zero \emph{location}, not zero order. This matters
experimentally, because a real interference pattern never has an exact node. Its finite
fringe visibility $V<1$ corresponds to some $\epsilon>0$, and the predicted divergence
appears not as a literal pole but as a finite sign-changing peak of height
$\hbar/2\epsilon$ that grows as $V\to1$. The peak height is a bounded proxy for proximity
to a true node, and the profile $-\hbar x/(x^2+\epsilon^2)$ is what a conjugate-pointer
readout (Sec.~\ref{sec:experiment}) resolves.

A finite visibility can arise in two physically distinct ways, and these must be kept
apart, because they differ in the \emph{current} channel even though they agree in the
osmotic one. \emph{Coherent} visibility loss is the pure-state model $\phi=x-i\epsilon$
used above, a zero displaced off the real axis. It regularizes the osmotic profile but
simultaneously produces the current-channel Lorentzian of part~(i), of height
$\hbar/\epsilon$ at the fringe center, the \emph{same} order as the osmotic peak
$\hbar/2\epsilon$. \emph{Incoherent} visibility loss, a mixed $\hat\rho$ such as a
partially distinguishable two-path mixture, leaves the phase of every pure component
smooth at the intensity minimum, so the current channel stays smooth. The osmotic profile
is meanwhile fixed by the recorded intensity alone, since
$\Ima(k_{x,w})=-\tfrac12\partial_x\ln\rho(x)$ holds for pure and mixed states alike, and
takes the same dispersive form with the same $\epsilon(V)$. The two are therefore
distinguished not by the osmotic profile but by whether the current channel carries the
companion spike. The smooth $\Rea(p_w)$ that Kocsis \emph{et al.}\ measured through the
dark fringes~\cite{kocsis2011} already points to the incoherent mechanism there, and
Sec.~\ref{sec:experiment} accordingly builds its quantitative prediction on the intensity
route.

\begin{remark}[Scope]
\label{rem:scope}
The results as proved carry three standing restrictions, collected here instead of being left
implicit. (i) They are one-dimensional. In $d>1$ the zero of $\phi^*\psi$ is
generically an isolated vortex about which $\Theta$ winds by $2\pi\ell$, so
$\nabla\Theta\sim\ell/r$ and the boundedness hypothesis of Theorem~\ref{thm:main} fails.
Both parts of $p_w$ then diverge, as
$p_w\to\hbar\ell\,\hat\varphi/r-i\hbar k\,\hat r/r$, with independent integer strengths
and asymptotically orthogonal directions. The Madelung split
\eqref{eq:re}--\eqref{eq:im} remains a pointwise identity
(Remark~\ref{rem:kinematic}), but the asymmetry asserted by Theorem~\ref{thm:main} is a
one-dimensional statement and does not survive. The node-inaccessibility argument changes
as well, the planar Laplacian contributing an extra $\nu/r$ so that the Bessel dimension
becomes $2k+2$ rather than $2k+1$. That regime is the subject of singular optics, where the codimension-two nodal geometry of
wave fields has been developed since Nye and Berry~\cite{nyeberry1974,dennis2009}, and the
two-state version of it is left to future work (Sec.~\ref{sec:caveats}).

(ii) They hold on the real-envelope class of Theorem~\ref{thm:main}, in which $\phi^*\psi$ admits the signed factorization
\eqref{eq:signedpolar}. Energy eigenstates, Hermite--Gaussian modes, and symmetric
interference patterns lie in it. The hypothesis costs the order of the node.
Within the class a \emph{simple} ($k=1$) real node is generic, one real condition $R=0$ in one
real variable, whereas $k\geq2$ requires $R=R'=\dots=R^{(k-1)}=0$ and is again
non-generic. In particular no Hermite--Gaussian mode or two-slit interference pattern
supplies one, since the zeros of $H_n$ are all simple. Results stated for general $k$ are
therefore exact but, for $k\geq2$, apply to engineered nodes and not to naturally occurring
ones. Additivity supplies the engineering: two simple zeros, one in each state, at the same
point give $k=2$ in the amplitude (Sec.~\ref{sec:additivity}).

(iii) The Hamiltonians are quadratic, which is used only to obtain a closed-form node
trajectory $\xnode(t)$ and an explicit Bernstein/Mehler construction to check against,
never for the kinematic Theorem~\ref{thm:main} itself (Remark~\ref{rem:kinematic}).

Restrictions (i) and (ii) are two readings of a single fact, not a double
standard. Isolated real zeros are codimension two, and hence non-generic, for an arbitrary
complex wavefunction. They are codimension one, and hence generic, within the
real-envelope class, which is exactly the class the physically central cases inhabit. The higher-codimension
nodal geometry of (i) is a genuinely different object, flagged as future work, not a gap in
the $1$D statement proved here.
\end{remark}

\begin{remark}[Relation to quantum Schr\"odinger bridges and to singular optics]
\label{rem:related}
The name quantum Schr\"odinger bridge covers four recent constructions, and what each does
has to be separated, because what the present results bound is the positivity
hypothesis they share, not any assertion made in them.

Pavon~\cite{pavon2003} works on configuration space and forms the complex drift
$v_q=\tfrac{\hbar}{im}\nabla\log\psi=v-iu$, which is \eqref{eq:re}--\eqref{eq:im} for a
\emph{single} state, and the post-measurement process there is the Nelson process of another
solution of the same Schr\"odinger equation. Movilla Miangolarra \emph{et
al.}~\cite{miangolarra2025} do have the two-state product structure, their
$\tilde\rho_\tau=\phi_\tau^{1/2}\hat\phi_\tau\phi_\tau^{1/2}$, but the dynamics there are
finite-dimensional Kraus maps, so no configuration space is available for a gradient to
live on, and their weak value is accordingly defined as a real part.
Ohzeki~\cite{ohzeki2026} makes the imaginary part of a weak value the optimal generator,
and is explicit that the full noncommutative stochastic bridge problem is deliberately set
aside there in favor of a coherent pure-state limit. The dynamics are accordingly unitary,
and his $\Ima A_w$ is the response of the post-selection probability to a coupling
parameter rather than a velocity on configuration space.
Aw and Sidajaya~\cite{aw2026} treat Petz recovery and prior hacking, with no weak value at
all. Table~\ref{tab:related} collects this.

The pattern in the last column is the substantive point. The constructions of
Refs.~\cite{pavon2003,miangolarra2025,aw2026}, independently of one another,
assume positivity of the endpoint data, and each notes what goes wrong without it.
Reference~\cite{aw2026} shows that a dephasing channel, which violates the condition, admits the
construction only on a strictly smaller set, which is the discrete counterpart of the
cell restriction in Corollary~\ref{cor:conditioning}. On configuration space with a
two-state amplitude the condition is not a proviso but has a locus of failure, the nodal
curve, and Theorem~\ref{thm:noflux} is what happens there. This paper therefore marks the
boundary of a framework, not a contradiction of a published claim.

The positivity hypothesis is older than these constructions, and so is the program of
relaxing it. Fortet~\cite{fortet1940} treated the one-dimensional Schr\"odinger system
for a continuous nonnegative kernel, and Blanchard, Garbaczewski, and
Olkiewicz~\cite{garbaczewski1994,blanchard1995} observed that strictly positive
Feynman--Kac kernels force the interpolating density to vanish only on the boundary of
the confining volume, and extended the problem to nonnegative kernels generated by
singular potentials, obtaining processes whose densities develop and destroy interior
zeros. Those results delimit, rather than overlap, the present ones, in three respects.
Their zeros are zeros of a single interpolating density, the Euclidean counterpart of
$\psi\bar\psi$, never of a two-state amplitude. Their persistent node carries no flux:
in their free example, built on the odd initial state $x\,e^{-x^2/4}$, the current
velocity vanishes at the node for all times, the $\dot M=0$ case in which
Proposition~\ref{prop:inaccessible} applies, and their conclusion there, an inaccessible
repelling barrier separating two non-communicating processes, is the stationary $k=1$
instance of it. And their node-destroying evolution places the zero at a single
instant, so no nodal curve exists on an open time interval, the hypothesis of
Theorem~\ref{thm:noflux} is never met, and mass crosses where the node was only after
the node has ceased to exist. Within their framework, wherever the node persists the
flux through it vanishes, and wherever the flux is nonzero the node does not persist.
Their extension and Theorem~\ref{thm:noflux} are therefore complementary rather than in
tension.

Relative to singular optics the difference is the same one. Berry~\cite{berry2009} and
Bliokh \emph{et al.}~\cite{bliokh2013} work with a single field, so their $\nabla\ln|E|$ is
the modulus gradient of one wave and their zeros are its own. Here the zeros may be
supplied by the post-selected state, what fails at them is a conditioning and not a
representation, and the discriminator of Sec.~\ref{sec:experiment} separates coherent from
incoherent visibility loss in the current channel, which the single-field reading does not.
A multi-photon classical field reproduces the mean profile but certifies nothing about
noncontextual models.
\end{remark}

\begin{remark}[A kinematic fact, not a dynamical one]
\label{rem:kinematic}
Nothing in Theorem~\ref{thm:main} refers to a Hamiltonian. It is a statement about the
Madelung decomposition \eqref{eq:re}--\eqref{eq:im} of two arbitrary normalizable states
$\psi,\phi$ at a fixed instant, and holds regardless of what dynamics generated them.
The restriction to quadratic Hamiltonians used elsewhere in this paper
(Secs.~\ref{sec:nodecaustic}--\ref{sec:numerics}, Appendices~\ref{app:bernstein}--\ref{app:beyond})
is needed only to obtain a closed-form trajectory $\xnode(t)$ and an explicit
Bernstein/Mehler construction to verify against, not for the theorem itself, which is
purely algebraic.
\end{remark}

\begin{remark}[Relation to contextuality]
\label{rem:contextuality}
Anomalous weak values require the overlap $\langle\phi|\psi\rangle$ to be small, i.e.\
proximity to $\partial\mathcal{D}$. They are contextual~\cite{pusey2014}. That the
obstruction lives in $\Ima(p_w)$ is consistent with the analysis of Kunjwal, Lostaglio,
and Pusey~\cite{kunjwal2019}, in which the imaginary part of the weak value and the
robustness of the contextuality witness are explicitly tied. Section~\ref{sec:experiment}
makes this quantitative, for an explicitly constructed effect. The
osmotic field is, up to the post-selection probability and a fixed normalization, the
imaginary Kirkwood--Dirac quasiprobability that violates the KLP inequality.
\end{remark}

\begin{remark}[Relation to Mori and Tsutsui]
\label{rem:moritsutsui}
Mori and Tsutsui~\cite{moritsutsui2015} observed a related pattern in a double-slit
setting, the real part of a weak value tracking a classical trajectory while the imaginary
part diverges under destructive interference, which they read as the particle and the wave
aspect of a weak trajectory. They write
$A_w=\langle\psi|A|\phi\rangle/\langle\psi|\phi\rangle$ with the roles of the two symbols
interchanged relative to the convention used here, and we quote their results in our
labelling. The parallel with the split \eqref{eq:re}--\eqref{eq:im} is close, but the two
divergences are different objects. Theirs is global. It occurs when the total transition
amplitude $K=\langle\psi|\phi(T)\rangle$ [their $K(0)$] vanishes, and is found by scanning
the post-selection parameter until it does. Ours is local. It occurs at fixed and
generically nonzero overlap, as a function of $x$, wherever $\phi^*\psi(x)$ has a spatial
node. One is a divergence in the post-selection knob and the other a divergence in the
transverse coordinate, so no single measurement can confuse them. Both can be present at
once, and which parameter is varied separates them. Mori and Tsutsui note themselves that
their weak-trajectory construction differs from the Wiseman and Kocsis current-velocity
line~\cite{wiseman2007,kocsis2011} that the present paper builds on and extends.
\end{remark}

\section{Node versus caustic, two distinct boundaries}
\label{sec:nodecaustic}

A second locus of divergence appears in the delta limit, and it is not a boundary of the
conditioning. In the
position-eigenstate (delta) limit, where $\phi^*(x,t)=\langle x_f|U(T,t)|x\rangle=K$ is
the forward propagator (Sec.~\ref{sec:madelung}), the oscillator's delta-limit field is
\begin{equation}
\Rea(p_w)/m=v_\psi-\dot x_{\text{cl}}(t),\qquad
\dot x_{\text{cl}}(t)=\omega\,\frac{x_f-x\cos\omega\tau}{\sin\omega\tau},
\label{eq:xdotcl}
\end{equation}
with $\tau=T-t$, where $\dot x_{\text{cl}}$ is the velocity of the classical trajectory constrained to
reach $x_f$ at time $T$ (Appendix~\ref{app:mehler}) and $v_\psi=p_\psi/m$ is the
pre-selected state's own current velocity, a spatial constant in the real-envelope class.
We set $v_\psi=0$ in what follows, the general case adding only that constant. The overall sign fixes the relation between $\Rea(p_w)$ and
a bridge drift in the sharpest available case. A bridge steering a diffusion toward $x_f$
drifts \emph{toward} it, and $-\dot x_{\text{cl}}$ points the other way. The reason is structural. Since $\phi$
propagates backward, its phase gradient enters
$\Rea(p_w)=\partial_x(S_\psi+S_\phi)$ with the sign of the time-reversed motion, as
Eq.~\eqref{eq:mehlerscore} makes explicit.
Proposition~\ref{prop:current} makes a related point at the level of processes, but
\eqref{eq:xdotcl} makes this one in a line, before any stochastic machinery is invoked.
The sign aside, \eqref{eq:xdotcl} also diverges, and on the entire time-slice
$\omega\tau=\pi$ rather than at any particular $x$.
This is the Mehler caustic, a semiclassical focusing singularity of the propagator derived
from the two-point action in Appendix~\ref{app:mehler}. It expresses the fact that the
boundary-value problem $x\to x_f$ over exactly half a period admits no finite-velocity
solution. It has nothing to do with a node of any wavefunction, and for a finite-width,
normalizable state it is not there at all. The overlap stays nonzero, and both
$\Rea(p_w)$ and $\Ima(p_w)$ remain finite through $\omega\tau=\pi$.

\begin{figure}[!htbp]
\centering
\includegraphics[width=\columnwidth]{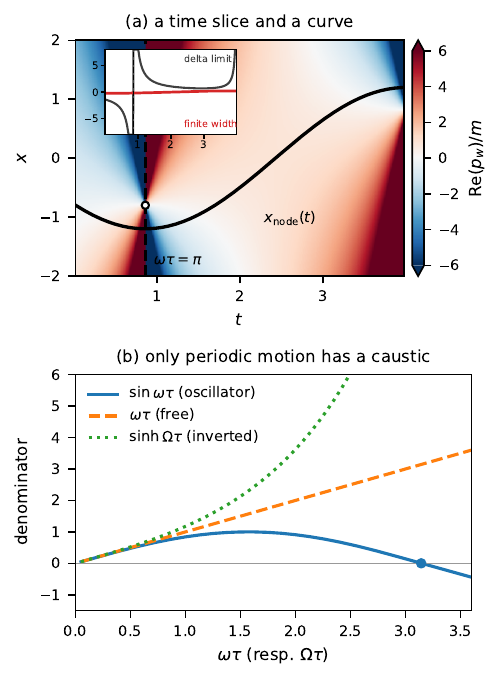}
\caption{Caustic versus node. \textbf{(a)} The delta-limit field
\eqref{eq:xdotcl} diverges along an entire time-slice $\omega\tau=\pi$ (dashed), the same
$t$ for every $x$ (with the single exception of $x=-x_f$, where numerator and denominator
vanish together, Proposition~\ref{prop:geometry}), whereas the osmotic boundary
$\partial\mathcal{D}=\{x=\xnode(t)\}$ (solid) is a curve. The two are geometrically
distinct objects in the same plane. Inset, at fixed $x$. The delta-limit field blows up at
the caustic while the finite-width (coherent) $\Rea(p_w)$ of \eqref{eq:cohRe} passes
through it smoothly, the overlap staying nonzero. \textbf{(b)} The caustic exists only for
periodic motion, the constrained-velocity
denominator is $\sin\omega\tau$ for the oscillator (vanishing at $\tau=n\pi/\omega$,
circles), but $\omega\tau$ for the free particle and $\sinh\Omega\tau$ for the inverted
oscillator (Appendix~\ref{app:beyond}), which never vanish. The node curve of panel~(a) is
present in all three cases.}
\label{fig:caustic}
\end{figure}

The two are easy to conflate on a field plot, both appearing as a line of singular
behavior in the $(x,t)$ plane. The following isolates the test that distinguishes them.

\begin{proposition}[Geometric diagnostic]
\label{prop:geometry}
\begin{enumerate}
\item[(i)] The delta-limit field \eqref{eq:xdotcl} diverges on the entire time-slice
$\omega\tau=n\pi$, $n\in\mathbb{Z}_{>0}$, independently of $x$, except at the single point
$x=(-1)^{n}x_f$.
\item[(ii)] The osmotic boundary $\partial\mathcal{D}=\{(x,t):\phi^*\psi=0\}$ is, at an
isolated simple zero, locally a curve $x=\xnode(t)$.
\item[(iii)] If $\psi,\phi$ are $C^1$ with $\phi^*\psi\neq0$ at $(x,t)$, both $\Rea(p_w)$
and $\Ima(p_w)$ are finite there, whatever the value of $\omega\tau$.
\end{enumerate}
\end{proposition}

\begin{proof}
(i) With $\tau=T-t$, \eqref{eq:xdotcl} reads
$\Rea(p_w)/m=-\omega(x_f-x\cos\omega\tau)/\sin\omega\tau$. The denominator vanishes
exactly on $\omega\tau\in\pi\mathbb{Z}$, a locus independent of $x$, and the numerator is
then $x_f-(-1)^{n}x$, which vanishes only at the single point $x=(-1)^{n}x_f$. Away from
that point the field diverges. At that point l'H\^opital's rule in $\tau$ gives the finite
limit $0$.
This is the conjugate point through which the classical family does pass
(Appendix~\ref{app:mehler}).

(ii) At a simple zero, $\phi^*\psi=0$ and $\partial_x(\phi^*\psi)\neq0$, so the implicit
function theorem gives a unique $C^1$ solution branch $x=\xnode(t)$ locally.

(iii) By \eqref{eq:pwsigned}, $p_w=-i\hbar\,\partial_x\ln(\phi^*\psi)$ is finite wherever
$\phi^*\psi\neq0$ and both factors are $C^1$, and neither $\omega$ nor $\tau$ appears. The
divergence of \eqref{eq:xdotcl} arises only from the delta limit
$|\phi_f\rangle\to|x_f\rangle$, in which $\phi^*$ becomes the propagator $K$ and the
Van Vleck prefactor degenerates (Appendix~\ref{app:mehler}). For a post-selection of
finite width $\sigma$ the overlap $\langle\phi|\psi\rangle$ stays nonzero through
$\omega\tau=\pi$, and the apparent singularity is washed out on the scale set by $\sigma$.
The caustic is in this sense a delta-limit artifact. The osmotic divergence, by contrast,
is the content of Theorem~\ref{thm:main} and persists for every such pair.
\end{proof}

The distinction becomes sharper once one leaves the oscillator. The delta-limit field
\eqref{eq:xdotcl} diverges where its denominator $\sin\omega\tau$ vanishes. But this
denominator is $\omega\tau$ for the free particle and $\sinh\Omega\tau$ for the inverted
oscillator (Appendix~\ref{app:beyond}), neither of which vanishes at any finite $\tau>0$.
The caustic is therefore not a generic feature of quadratic dynamics, but is special to
\emph{periodic} classical motion, recurring every half period only for the ordinary
oscillator (Fig.~\ref{fig:caustic}b). The node curve $\xnode(t)$ persists for all three,
and node and caustic are geometrically distinct (Proposition~\ref{prop:geometry}) and need
not both be present. Section~\ref{sec:numerics} shows the node and the regularized caustic
side by side.

\section{Closed-form scan across the node and sampled trajectories}
\label{sec:numerics}

We check all of this for the oscillator ($\hbar=m=\omega=1$) with total interval $T=4$,
chosen so that $\pi<\omega T<2\pi$ puts exactly one Mehler caustic ($\omega\tau=\pi$)
inside the window. The pre-selection is a coherent state of phase-space center
$(q_0,p_0)=(0.6,0.4)$. The post-selection is either a nodeless coherent state of center
$(q_f,p_f)=(-0.8,0)$ at time $T$, or a displaced first-excited (Hermite--Gaussian)
state whose center, and with it the single node $\xnode(t)$, sits at
$(q,p)=(1.2,0)$ at time $T$ and follows the unconditioned classical trajectory when
propagated back. Both fields are evaluated by central differences of $\ln(\phi^*\psi)$
with step $10^{-6}$, and checked against the closed-form expressions of
Appendices~\ref{app:bernstein} and~\ref{app:mehler} where available. The maps of
Figs.~\ref{fig:flagship}--\ref{fig:control} use a $700\times400$ grid over
$(x,t)\in[-3,3]\times[0.02,T-0.02]$, i.e.\ $\Delta x\simeq8.6\times10^{-3}$. The scan of
Table~\ref{tab:scan} is \emph{not} taken from that grid, its smallest offset
$x-\xnode=3\times10^{-3}$ lying below the grid spacing. It is evaluated directly at
the listed points $\xnode(t)+(x-\xnode)$, with $\xnode(t)$ obtained from the classical
trajectory of the post-selected center. At the closest offset the differencing step is
still smaller than the distance to the pole by a factor $3\times10^{3}$, and the
associated truncation error, $O(h^2/y^3)\sim10^{-5}$, is negligible against the tabulated
values.

For this configuration the scan admits a closed form, which we give here because
it makes the check exact instead of graphical. Both states lie in the real-envelope class
of Remark~\ref{rem:scope}, a real envelope times a spatially uniform plane-wave phase, so
$\phi^*\psi$ has a phase strictly linear in $x$. Writing $q_\psi(t),p_\psi(t)$ for
the pre-selected coherent center and $\partial_xS_\phi(t)$ for the uniform phase gradient
of the post-selected Hermite--Gaussian, whose node sits at $\xnode(t)$, one obtains
(verified symbolically, and in the units $\hbar=m=\omega=1$ of this section)
\begin{equation}
\begin{aligned}
\Rea(p_w)&=p_\psi+\partial_xS_\phi,\\
\Ima(p_w)&=-\frac{1}{x-\xnode}+(x-\xnode)+(x-q_\psi).
\end{aligned}
\label{eq:hgclosed}
\end{equation}
This is Theorem~\ref{thm:main} realized at $k=1$ in closed form. The pole sits in the osmotic
channel alone, and $\Rea(p_w)$ is $x$-independent across $\partial\mathcal{D}$, being the sum of the two
uniform phase gradients. In the convention of Sec.~\ref{sec:madelung} that sum equals
$p_\psi-m\dot q_\phi$, since $\partial_xS_\phi$ is minus the momentum of the post-selected
center. We
report the scan below at the precision needed to make the agreement with
\eqref{eq:hgclosed} visible. The numerics of this section are therefore a verification of
an exact expression, not an independent discovery. Their role is to confirm that
the differencing scheme used for the field maps, where no closed form is available,
resolves the pole correctly.

The clearest number is a scan across the node (Table~\ref{tab:scan}). At two different
times, shrinking the distance to $\xnode$ by a factor of ten increases $|\Ima(p_w)|$ by
factors of $9.5$ and $10.1$ respectively across the full scan from $0.03$ to $0.003$,
which is ten up to the smooth additive remainder in \eqref{eq:hgclosed}, while
$\Rea(p_w)$ does not move at all. This is the
$1/(x-\xnode)$ scaling of Theorem~\ref{thm:main} seen directly in the numerical scan.

\begin{table}[t]
\caption{Node-neighborhood scan, displaced Hermite--Gaussian post-selection (parameters
as specified at the start of this section). Evaluated directly at the listed offsets, not
on the figure grid. Every entry reproduces the closed form \eqref{eq:hgclosed} to all
digits shown. $\Ima(p_w)$ carries the $1/(x-\xnode)$ pole, while $\Rea(p_w)$ is independent of $x$ at each $t$.}
\label{tab:scan}
\begin{ruledtabular}
\begin{tabular}{ccrr}
$t$ & $x-\xnode$ & $\Rea(p_w)$ & $\Ima(p_w)$\\
\hline
$0.8$ & $0.03$  & $-0.081682$ & $-35.1763$\\
$0.8$ & $0.01$  & $-0.081682$ & $-101.8829$\\
$0.8$ & $0.003$ & $-0.081682$ & $-335.2303$\\
$2.5$ & $0.03$  & $-1.876535$ & $-32.9472$\\
$2.5$ & $0.01$  & $-1.876535$ & $-99.6538$\\
$2.5$ & $0.003$ & $-1.876535$ & $-333.0012$\\
\end{tabular}
\end{ruledtabular}
\end{table}

The same content is visible over the whole $(x,t)$ plane in Fig.~\ref{fig:flagship}. For
nodeless post-selection $\Ima(p_w)$ is smooth everywhere, including straight through the
Mehler caustic $\omega\tau=\pi$, the locus on which the delta-limit field
\eqref{eq:xdotcl} diverges (Proposition~\ref{prop:geometry}). For nodal
post-selection the same field diverges cleanly along the curve $\xnode(t)$, and nowhere
else. Figure~\ref{fig:control} supplies the control. There $\Rea(p_w)/m$, computed for the
identical nodal post-selection, is smooth across that same curve.

\begin{figure}[!htbp]
\centering
\includegraphics[width=\columnwidth]{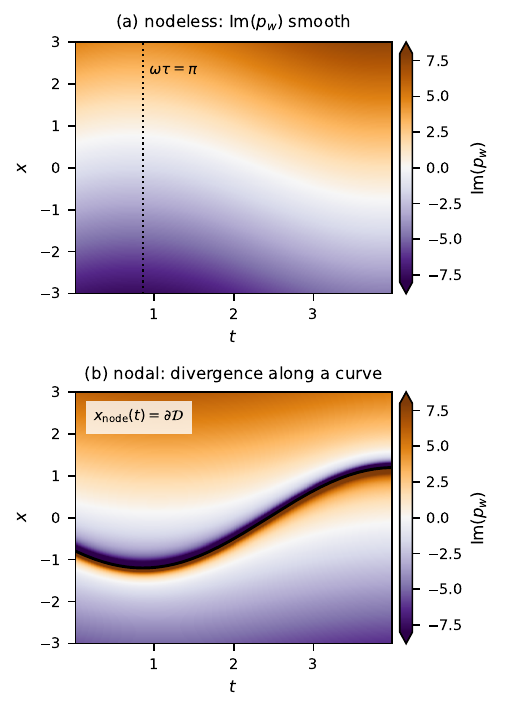}
\caption{Osmotic momentum $\Ima(p_w)$ for the oscillator (parameters in text). (a) Nodeless post-selection, smooth. The Mehler caustic $\omega\tau=\pi$ (dotted) is regularized and is \emph{not}
$\partial\mathcal{D}$. (b) Nodal post-selection, $\Ima(p_w)$ diverging as
$1/(x-\xnode)$, sign-changing, along the curve $\xnode(t)$ (black), the genuine
$\partial\mathcal{D}$.}
\label{fig:flagship}
\end{figure}

\begin{figure}[!htbp]
\centering
\includegraphics[width=\columnwidth]{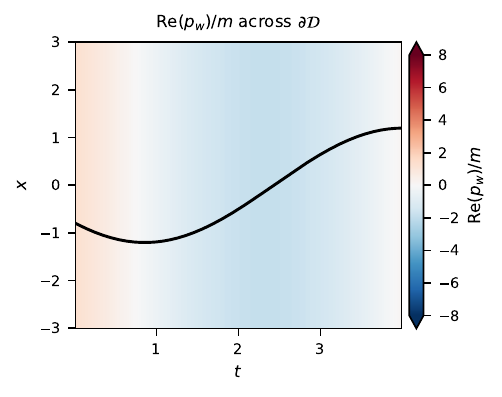}
\caption{Control. The current velocity $\Rea(p_w)/m$ for the same nodal post-selection
stays finite across $\partial\mathcal{D}$ (black). For this state class it is in fact
exactly independent of $x$ at each $t$ [Eq.~\eqref{eq:hgclosed}], so the panel is a
visualization of an exact identity. The divergence visible in
Fig.~\ref{fig:flagship}(b) is absent here at any resolution, which is what entitles one
to call it osmotic.}
\label{fig:control}
\end{figure}

The same sign-change mechanism appears, in a weaker form, in the discrete pre- and
post-selection paradoxes, where an anomalous weak value arises where the conditioning
weight $\phi^*\psi$ fails to be a positive overlap. Neither the three-box problem nor
Hardy's paradox has a node, and neither is an instance of Theorem~\ref{thm:main}. Both are
worked out, the exact three-well convergence and the Hardy occupation values, in
Appendix~\ref{app:discrete}.

Finally, Fig.~\ref{fig:traj} returns to $\mathcal{D}$ itself and makes
Proposition~\ref{prop:current} visible on sampled paths. Because
Proposition~\ref{prop:current} forbids any diffusion that carries the product marginal
$\rho_\psi\rho_\phi(\cdot,t)$ with current velocity $\Rea(p_w)/m$, we sample the
demonstration diffusion that \emph{does} exist. Take coherent pre-/post-selection with
$\omega T<\pi$, so the run stays nodeless and on $\mathcal{D}$, and let
$\bar v(t)=\Rea(p_w)/m$ be the coherent-class field [computed numerically from the phase
gradient of $\phi^*\psi$, checked against the closed form~\eqref{eq:cohRe} and for
spatial uniformity]. Define the marginal $\bar\rho(\cdot,t)$ as a Gaussian of fixed
coherent-state width $\sigma^2=\hbar/2m\omega$ whose center $\mu(t)$ is transported by
the continuity equation at $\bar v$, i.e.\ $\dot\mu=\bar v$ with $\mu(0)=q_\psi(0)$, and
integrate the It\^o equation~\eqref{eq:ito} with drift
$b=\bar v+\nu\,\partial_x\ln\bar\rho=\bar v-\omega\,[x-\mu(t)]$ (an Ornstein--Uhlenbeck drift
with moving center, for which $\sigma^2$ is the stationary variance), initial points
drawn from $\bar\rho(\cdot,0)$.

This is a genuine diffusion whose Nelson current
velocity is exactly $\bar v$ and whose marginal remains $\bar\rho$ for all $t$. By
Proposition~\ref{prop:current}, $\bar\rho(\cdot,t)$ cannot coincide with
$\rho_\psi\rho_\phi(\cdot,t)$. The product density has width $\sigma^2/2$ and a center
transporting at $\tfrac12(\dot q_\psi+\dot q_\phi)$, which $\dot\mu=\bar v$ misses by
$\tfrac12(\dot q_\psi-3\dot q_\phi)$. The ensemble is therefore a construction exhibiting
the current and osmotic split of \eqref{eq:ito}, a demonstration of
Proposition~\ref{prop:current} rather than the physical real-time conditioned ensemble.

Trajectories are generated by Euler--Maruyama [$\hbar=m=\omega=1$, $\nu=\tfrac12$;
$(\alpha,\beta)=(1,0.8)$, $T=2.5$; $N=3\times10^5$ paths, $K=1200$ steps,
$\Delta t\approx2.1\times10^{-3}$], and the local drifts below are estimated from the
forward and backward increments of the same paths,
$b_\pm(x)=\pm\langle X_{t\pm\Delta t}-X_t\,|\,X_t=x\rangle/\Delta t$, binned in $x$
(22 bins) at the fixed slice $t^*=0.55\,T=1.375$. The trajectory ensemble is preceded by
two closed-form checks on the
sampling. $\Rea(p_w)$ evaluated from $\phi^*\psi$ is spatially uniform to
$5.6\times10^{-11}$ over $x\in[-0.7,0.9]$, and agrees with \eqref{eq:cohRe} to
$2.2\times10^{-16}$, i.e.\ to machine precision. Panel~(a) is a sanity
check. The empirical mean tracks the analytic $\int\bar v\,dt$ along an oscillatory curve
and not the straight line of a point (delta-post-selected) bridge, with
$\max_t|\tfrac{d}{dt}E[X_t]-\bar v(t)|=0.098$ over the interior and sampled variance $0.501$
against the target $\sigma^2=0.5$. Panel~(b) is the substantive test. At $t^*$, where
$\bar v=-1.703$, the local It\^o drift
$b_+(x)$ is sloped, the Nelson current velocity
$\tfrac12(b_++b_-)$ is flat and equal to $\bar v=\Rea(p_w)/m$, and their difference is the
osmotic velocity $\nu\,\partial_x\ln\bar\rho$. Flatness here is a statement about scales,
not an exact one. Binned over the window $\mu(t^*)\pm1.6$ in $x$, the current-velocity
estimator varies by $0.84$ while the It\^o drift varies by $3.37$, a factor of four, with the residual
scatter of both consistent with the finite bin occupancy in the tails. Post-selection does not remove noise here. The ensemble mean realizes the current
velocity while the osmotic half of the drift stays hidden from it, which is
Proposition~\ref{prop:current} and is why a picture built from $\Rea(p_w)$ alone cannot
see the boundary (Sec.~\ref{sec:caveats}).

\begin{figure}[!htbp]
\centering
\includegraphics[width=\columnwidth]{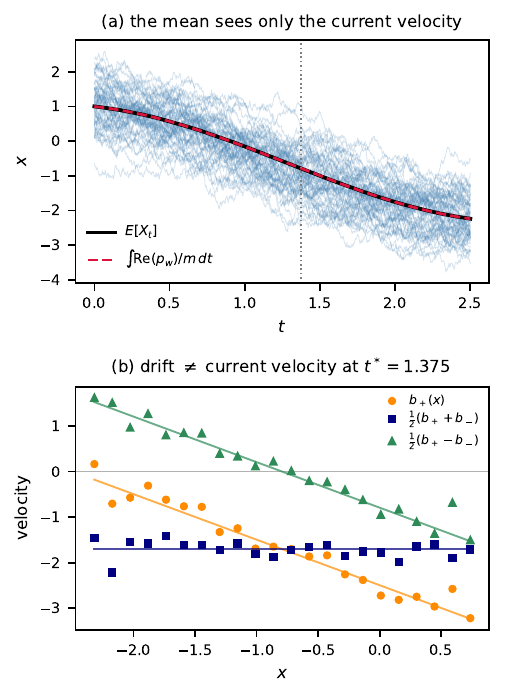}
\caption{Proposition~\ref{prop:current} on sampled trajectories. Coherent (nodeless)
pre-/post-selection on $\mathcal{D}$, oscillator with $\omega T<\pi$ (no caustic). The
sampled process is the demonstration diffusion of the text, with drift
$b=\bar v+\nu\,\partial_x\ln\bar\rho$, $\bar v(t)=\Rea(p_w)/m$ [Eq.~\eqref{eq:cohRe}], and
$\bar\rho$ a fixed-width Gaussian (coherent width $\sigma^2=\hbar/2m\omega$) whose
center is transported at $\bar v$. By
Proposition~\ref{prop:current} its marginal is not the product density
$\rho_\psi\rho_\phi(\cdot,t)$.
\textbf{(a)} The ensemble
mean $E[X_t]$ (black) tracks the independently computed field
$\int\!\bar v\,dt$ (dashed) along an oscillatory, not a straight, curve.
\textbf{(b)} From the \emph{same} paths at fixed $t$, the local It\^o drift $b_+(x)$
(sloped) differs from the Nelson current velocity $\tfrac12(b_++b_-)$ (flat, equal to
$\Rea(p_w)/m$) by exactly the osmotic velocity $\tfrac12(b_+-b_-)$. The mean of panel~(a)
sees only the flat current velocity. The drift and its osmotic part are recovered only by
the two-sided estimator, illustrating $b=v+u$ with $\langle u\rangle=0$.}
\label{fig:traj}
\end{figure}

\section{Experimental connections and a same-apparatus prediction}
\label{sec:experiment}

The current-velocity field $\Rea(p_w)/m$ has already
been measured, and our result makes a specific, so-far-untested prediction for the field
that carries the boundary.

Kocsis \emph{et al.}~\cite{kocsis2011} reconstructed average single-photon trajectories
in a two-slit interferometer by weakly measuring photon momentum and post-selecting on a
strong position measurement in a sequence of planes, precisely an operational
determination of $\Rea(p_w)/m$ for a position-post-selected ensemble (the Wiseman
identification~\cite{wiseman2007}). The same reconstruction was later carried out for one
of a pair of entangled photons~\cite{mahler2016}, which extends the technique but leaves
the single-particle geometry used below unchanged. Their trajectories are smooth curves
that weave
through the interference pattern, and they do not diverge at the dark
fringes, which are nodes of the interfering wavefunction. That is required by Theorem~\ref{thm:main}
on the current-velocity side, already sitting in published data. The fringes are points
of $\partial\mathcal{D}$, and $\Rea(p_w)$ passes through them finitely.

The same theorem requires the osmotic momentum $\Ima(p_w)$ to diverge there, as
$1/(x-\xnode)$ with a sign change across the fringe. This channel was not recorded in that
run. Kocsis \emph{et al.}\ projected the polarization pointer onto the circular basis,
which returns the real part alone. (In the photonic realization the transverse coordinate
plays the role of position under the paraxial field~$\cong$~2D-Schr\"odinger
correspondence, and the weak pointer is the photon polarization. The imaginary part is
carried by a different polarization component, identified explicitly below.) The imaginary
part of the local momentum is not itself an unmeasured quantity. In classical optics it has
been accessed through the gradient force on probe particles~\cite{bliokh2013}, and Bliokh
\emph{et al.}~\cite{bliokh2013} pointed out in their classical-optics reading of the
Kocsis run that the same apparatus would return $\Ima(p_w)$ if the Stokes pointer
orthogonal to the one used there were recorded instead. We take that observation as the
starting point of the proposal below. The computation is the same in the two settings, and
the difference lies in what the result can certify. A multi-photon classical field
reproduces the mean values but not the single-quantum operational statistics that the
contextuality inequality \eqref{eq:klp} constrains, so the same profile obtained
classically certifies nothing about noncontextual models. For the position post-selection
used in the Kocsis geometry the two readings return the same number, since
$\phi^*\psi\propto\psi$ there. Resolving the osmotic channel across an interference node has not been done in either the
classical or the post-selected quantum setting, and that is where
Theorem~\ref{thm:main} makes its prediction. Repeating
the experiment with the conjugate projection therefore turns their remark into a
quantitative test. It should reveal a divergence localized exactly on the interference
nodes, with the sign-changing profile of Eq.~\eqref{eq:imdiv} and
Fig.~\ref{fig:flagship}(b), on the scale $\epsilon(V)$ fixed by the measured fringe
visibility. That both pointer channels are
separately accessible is established experimentally. Lundeen \emph{et
al.}~\cite{lundeen2011} already extracted a wavefunction's amplitude and phase from the
real and imaginary parts of a weak value of position, post-selected in momentum, in
exactly this way. The Madelung split of Sec.~\ref{sec:madelung}, read in this light, is a
statement about which of two independently measurable channels carries the singularity.

In this position-post-selected geometry the
osmotic profile is not an independent field. For position post-selection the weak
momentum is $p_w=-i\hbar\,\partial_x\ln\psi(x)$, whose imaginary part, the
position-post-selected case of the split \eqref{eq:im}, gives
$\Ima(k_{x,w})=-\tfrac12\,\partial_x\ln\rho(x)$, set by the interference intensity
$\rho(x)$ already recorded on the CCD, so its shape can be read from the existing data by a
logarithmic derivative. The proposed measurement is therefore operational in the sense of Lundeen \emph{et al.}~\cite{lundeen2011}. The question is whether the
conjugate linear-polarization channel, read out independently, reproduces that profile with
its sign change and $1/(x-\xnode)$ growth, or departs from it near the node. A departure
is the informative outcome, and the divergence is what makes it sharply testable.

\subsection{A test at zero cost on data already recorded}
The first test costs nothing, and it lies in the channel Kocsis \emph{et al.}\ already
recorded.

The setting is the paraxial mapping, where the propagation axis $z$ plays the role of time,
the transverse coordinate $x$ that of position, and $\hbar k_x$ that of momentum, with
$k_0=2\pi/\lambda$ the optical wavenumber. (The symbol $k$ stays reserved for the nodal
order of Theorem~\ref{thm:main}, which reappears in the rare-event counting further down.)
Thus $\Rea(k_{x,w})/k_0$ is the trajectory slope Kocsis reconstructed~\cite{kocsis2011} and
$\Ima(k_{x,w})/k_0$ is the osmotic partner they did not record.

A real interference pattern has no exact node. Expanding the recorded intensity
$\rho\propto1-V\cos[2\pi(x-\xnode)/\Lambda]$ about its minimum, for local fringe spacing
$\Lambda$ and visibility $V$, displaces the zero off the real axis by
\begin{equation}
\epsilon=\frac{\Lambda}{\pi}\sqrt{\frac{1-V}{2V}},
\label{eq:epsilonV}
\end{equation}
which is the $\epsilon$ of Proposition~\ref{prop:genericity} and tends to zero, a true
node, as $V\to1$. For the reported geometry~\cite{kocsis2011}, $\lambda=943$~nm and a beam
separation $d=4.69$~mm at $z=5.6$~m give $\Lambda=\lambda z/d\simeq1.13$~mm, hence
$\epsilon\approx26$--$85~\mu$m over $V=0.99$--$0.90$. That run used a three-lens
cylindrical imaging system, so $z$ is an effective propagation distance. We therefore read
$\Lambda$ off the published intensity pattern of their Fig.~2(C), and find it consistent
with the thin-lens value.

The discriminator follows from Proposition~\ref{prop:genericity}, whose two mechanisms for
a finite visibility agree in the osmotic channel and differ in the current one. Coherent
visibility loss, an off-axis zero, must place a
Lorentzian of width $\epsilon$ and peak
$|\Rea(k_{x,w})|/k_0=(\epsilon k_0)^{-1}\approx1.8\times10^{-3}$ ($V=0.90$) to
$5.9\times10^{-3}$ ($V=0.99$) at the fringe center, above the ${\sim}10^{-3}$
trajectory-slope scale already resolved, whereas incoherent loss predicts a smooth current
channel there. The published smoothness of the reconstructed trajectories through the dark
fringes is thus already qualitative evidence for the incoherent mechanism, and a dedicated
scan of the recorded position-pointer data across a dark fringe turns it into a
quantitative bound on any coherent (off-axis-zero) admixture. No change to the apparatus,
no new exposure, and no new run. The data exist.

The same remark makes Theorem~\ref{thm:noflux} falsifiable. Its hypothesis is
$\dot M\neq0$, and $M(t)$ is built from two ordinary intensity images. The factor
$\rho_\psi(\cdot,t)$ is what Kocsis \emph{et al.}\ already record at each imaging plane. The
factor $\rho_\phi(\cdot,t)$ is the intensity of the post-selected mode propagated the same
distance, which for a real-envelope $\phi_f$ is the same mode run in the reverse direction
and requires no interference with the first. Multiplying the two images pointwise and
normalizing gives $\rho_{\text{eff}}$, and integrating up to the fringe minimum gives $M$.
It is enough to image two planes straddling the node's motion. If $M$ differs between
them, then by
Theorem~\ref{thm:noflux} no diffusion of constant diffusion coefficient reproduces both
marginals while staying off the nodal curve, which such a diffusion reaches with positive
probability, and the
conditioned-diffusion representation is excluded for that pair by measurement and not
by calculation.

\subsection{The osmotic profile across a dark fringe}
What Theorem~\ref{thm:main} says about the channel that was \emph{not} recorded is the
following. Near a first-order ($k=1$) dark fringe, a node of the interfering field at
$\xnode$, Eq.~\eqref{eq:imdiv} with Proposition~\ref{prop:genericity} gives
\begin{equation}
\Ima(k_{x,w})=-\frac{x-\xnode}{(x-\xnode)^2+\epsilon^2},
\label{eq:kxprediction}
\end{equation}
a sign-changing dispersive profile with extrema $\mp1/(2\epsilon)$ at
$x-\xnode=\pm\epsilon$, with $\epsilon(V)$ given by \eqref{eq:epsilonV}. This is the
intensity route, mechanism~(ii) of Proposition~\ref{prop:genericity}:
$\Ima(k_{x,w})=-\tfrac12\partial_x\ln\rho(x)$ holds for pure and mixed $\hat\rho$ alike,
and expanding the recorded intensity about its minimum gives \eqref{eq:kxprediction}
directly. It is therefore consistent with the smooth $\Rea(k_{x,w})$ Kocsis measured
through the dark fringes. The conjugate-channel readout tests whether the
independently measured pointer rotation reproduces that profile, and a departure has
identifiable candidate causes. The following three lie within reach of the same data.
The weak-coupling expansion behind \eqref{eq:qubitpointer} is first order in
$\zeta\Ima(k_{x,w})/k_0$, which reaches $1.1$~rad at $V=0.99$, so a shortfall growing with
visibility would indicate saturation of $\tanh\varphi_I$ rather than new physics. Residual
which-path coupling in the calcite would add a current-channel admixture, showing up as an
asymmetry between the two extrema at $x-\xnode=\pm\epsilon$. A pointer nonlinearity would
break the predicted proportionality between the rotation and $\partial_x\ln\rho$ while
leaving its zero crossing in place. Each has a distinct signature in the profile, which is
what makes the comparison a test with a null hypothesis and not a redisplay of the
intensity.

The peak $|\Ima(k_{x,w})|/k_0$ runs from
$8.8\times10^{-4}$ at $V=0.90$ to $2.9\times10^{-3}$ at $V=0.99$
(Fig.~\ref{fig:prediction}), the \emph{same} order as the current-velocity signal
$\Rea(k_{x,w})/k_0\lesssim10^{-3}$ they already resolved. With their calcite coupling
$\zeta=373.5\pm3.4$, the dimensionless coefficient of the linear birefringent phase
shift $\varphi(k_x)=\zeta\,k_x/k_0+\varphi_0$ that their calcite imparts to the polarization
pointer~\cite{kocsis2011}, this is a pointer rotation
$\zeta\,\Ima(k_{x,w})/k_0\approx0.33$--$1.1$~rad,
well above their demonstrated sensitivity. The upper end (as $V\to1$) exceeds the linear
weak regime at this fixed coupling. There one reduces $g$ to stay within it, and the
per-event signal saturates at order unity (the adaptive-$g$ argument below). The extrema,
at the $\epsilon\approx26$--$85~\mu$m quoted above, are separated by
$2\epsilon\approx2$--$7$ CCD pixels ($26~\mu$m), so the dispersive feature is resolved
across the range and is pixel-limited only at $V=0.99$.
The one genuine cost is the vanishing occupancy of the node region. At the readout extrema
$x-\xnode=\pm\epsilon$ the density is
suppressed by $\rho(\epsilon)/\rho_{\max}=2(1-V)/(1+V)\approx0.1$--$0.01$, so bringing a
node-neighborhood pixel to the shot-noise floor of a bright fringe requires
$\rho_{\max}/\rho(\epsilon)\approx10$--$10^2$ times the per-plane exposure. In absolute
terms this is minutes to hours, not months.

Kocsis's published run detected
${\sim}3\times10^4$ photons per $15$~s CCD exposure across the whole transverse
pattern~\cite{kocsis2011}. Because the cooled CCD images every $x$ position in parallel,
the overhead is an exposure-time multiplier and not a serial per-pixel scan, giving of
order $10^2$--$10^3$~s of additional integration per imaging plane, roughly two and a
half minutes at $V=0.90$ rising to ${\sim}25$~min at $V=0.99$. Their trajectory
reconstruction used $41$ imaging planes, so a full multi-plane scan costs about
$1.7$~h at $V=0.90$ and about $17$~h at $V=0.99$, a long acquisition, spread over a few
sessions at the high-visibility end, but on the same apparatus with the same optics, not a
different experiment. A single plane through a dark fringe, enough to test the
sign-changing profile itself, is a matter of minutes. Operationally, only the pointer projection changes, and the following computation fixes
which one. Three coupling strengths appear below, the dimensionful $g$ of a generic weak
measurement, the calcite parameter $\zeta$, and the dimensionless $\theta$ of
Proposition~\ref{prop:qubitklp}, related by $\theta=\zeta/2k_0$ as shown in
Appendix~\ref{app:qubitklp}.

The transverse field in the two-slit run is two-dimensional while
Theorem~\ref{thm:main} is one-dimensional, and the gap is closed by the geometry rather
than by assumption. The cylindrical optics of Ref.~\cite{kocsis2011} make the field
separable, $\psi(x,y)\simeq A(x)G(y)$ with $G$ a nodeless Gaussian carrying no
interference structure, so a zero of $\psi$ requires $A(x)=0$ and the nodal set is a line
$x=\xnode$ rather than an isolated point. Separability alone would not be enough, since a
complex $A$ would still require $\Rea A=\Ima A=0$. The second ingredient is that a
symmetric two-slit amplitude is real up to a common phase, which is the real-envelope class
of Remark~\ref{rem:scope}(ii). The nodal set is therefore codimension one in the transverse
plane, $\phi^*\psi$ admits the signed factorization \eqref{eq:signedpolar} with bounded
$\partial_x\Theta$, and the hypothesis of
Theorem~\ref{thm:main} is met. Unequal slit amplitudes displace the zero off the real axis
by the $\epsilon$ of Proposition~\ref{prop:genericity}, which is the finite-visibility case
already treated in \eqref{eq:kxprediction}. Genuine codimension-two nodes, the vortices of
singular optics~\cite{nyeberry1974,dennis2009}, arise only when both conditions fail, and
Remark~\ref{rem:scope}(i) records what changes there.

Which polarization component carries the osmotic part has to be settled explicitly.
Because the Kocsis pointer is a qubit rather than a continuous von Neumann pointer, the
assignment of $\Rea$ and $\Ima$ to pointer channels has to be made explicitly and cannot be
imported from the continuous-pointer result~\cite{jozsa2007}. The calcite imparts a
$\sigma_z$-type phase in the $\{|H\rangle,|V\rangle\}$ basis,
$\varphi(k_x)=\zeta k_x/k_0+\varphi_0$, to the initial pointer state
$|D\rangle=(|H\rangle+|V\rangle)/\sqrt2$. Writing $\varphi_w=\zeta (k_x)_w/k_0$ for the
(complex) weak phase, the post-selected pointer state is
$\propto e^{-i\varphi_w/2}|H\rangle+e^{+i\varphi_w/2}|V\rangle$, and with
$\varphi_w=\varphi_R+i\varphi_I$ a two-line computation (verified symbolically) gives
\begin{equation}
\langle\sigma_y\rangle=\frac{\sin\varphi_R}{\cosh\varphi_I},\quad
\langle\sigma_z\rangle=\tanh\varphi_I,\quad
\langle\sigma_x\rangle=\frac{\cos\varphi_R}{\cosh\varphi_I}.
\label{eq:qubitpointer}
\end{equation}
To leading order $\langle\sigma_y\rangle\simeq\varphi_R$ and
$\langle\sigma_z\rangle\simeq\varphi_I$, whereas
$\langle\sigma_x\rangle\simeq1-\tfrac12(\varphi_R^2+\varphi_I^2)$ is second order in
both. The real part is therefore read in the circular basis, which is the projection
Kocsis performed, and the osmotic part in the $\{|H\rangle,|V\rangle\}$ \emph{linear}
basis. The latter is obtained on the same bench by removing the quarter waveplate in front
of the lens system and leaving the beam displacer to separate $H$ from $V$. The other
linear basis, $\{|D\rangle,|A\rangle\}$, carries no first-order signal in either channel
and is \emph{not} the readout to use. The calcite, imaging optics, and post-selection are
unchanged, exactly as Lundeen \emph{et al.}~\cite{lundeen2011} separated the two parts of
a weak value on one apparatus. Equation~\eqref{eq:qubitpointer} also settles the
saturation caveat above. The readout is $\tanh\varphi_I$ rather than $\varphi_I$, so the
$\varphi_I\approx1.1$ quoted at $V=0.99$ corresponds to
$\langle\sigma_z\rangle\approx0.80$, already well into saturation. That is the regime the
adaptive-coupling argument below addresses.

\begin{figure}[!htbp]
\centering
\includegraphics[width=\columnwidth]{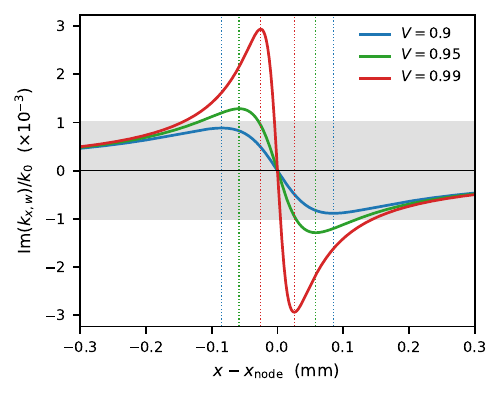}
\caption{Predicted osmotic (imaginary) transverse-momentum weak value
$\Ima(k_{x,w})/k_0$ across a dark fringe, i.e.\ Eq.~\eqref{eq:kxprediction}, which
gives $\Ima(k_{x,w})$, of dimension inverse length, divided by $k_0=2\pi/\lambda$, for
the Kocsis geometry~\cite{kocsis2011} at $z=5.6$~m ($\Lambda\simeq1.13$~mm) and three visibilities
$V$. Dotted lines mark the extrema at $x-\xnode=\pm\epsilon$. The grey band is the
$\pm10^{-3}$ scale of the current-velocity signal $\Rea(k_{x,w})/k_0$ already measured. The
imaginary signal is of the same order and sharpens toward a true $1/(x-\xnode)$
pole as $V\to1$.}
\label{fig:prediction}
\end{figure}

Resolving the divergence is less demanding than it appears. A weak
value is only a good linear estimate of the pointer shift when the coupling strength $g$
satisfies $g\lesssim\hbar/|p_w|$~\cite{dressel2014}. Since $|\Ima(p_w)|\sim\hbar/\delta x$
grows as the node is approached ($\delta x$ the distance to $\xnode$), the largest
coupling compatible with the weak regime shrinks correspondingly,
$g_{\max}\sim\delta x$. The resulting optimal signal,
$g_{\max}|\Ima(p_w)|/\hbar\sim\delta x\cdot(\hbar/\delta x)/\hbar=O(1)$, is therefore
independent of $\delta x$. Approaching the node does not demand resolving an
ever-growing shift, only adaptively shrinking $g$ to stay in the valid regime while the
achievable signal saturates at a fixed, order-unity scale. That order-unity signal is,
however, per \emph{post-selected} event, and post-selection onto a position window of
width ${\sim}\delta x$ about the node succeeds only with probability
${\sim}\rho\,\delta x\sim(\delta x)^{2k+1}$, since the density goes as
$\rho\sim(\delta x)^{2k}$ there (verified numerically). Resolving the osmotic profile down
to distance $\delta x$ therefore costs $N_{\text{inc}}\sim N_{\text{ps}}\,(\delta
x)^{-(2k+1)}$ incident photons for $N_{\text{ps}}$ retained events. This rare-event
overhead is real, but it is not special to the weak-value scheme. It is the vanishing
occupancy of a neighborhood of a node, and any method resolving a comparably sharp
feature at that location, weak-value or conventional, pays the same $(\delta
x)^{-(2k+1)}$. The divergence adds no penalty of its own. The per-event precision carries
no penalty either. Under ideal, shot-noise-limited detection, Ferrie and Combes proved that
weak-value amplification confers no fundamental precision advantage over conventional
measurement in that regime~\cite{ferriecombes2014}, so the divergence predicted here
carries no exotic metrological requirement. The practical case for a weak-value-based
readout is instead its established robustness to technical noise and detector
imperfections~\cite{jordan2014}, precisely the regime of real single-photon experiments
such as Kocsis's.

The same boundary has an operational face that has also been probed directly. Piacentini
\emph{et al.}~\cite{piacentini2016} realized Pusey's contextuality
test~\cite{pusey2014} in the laboratory, connecting anomalous weak values to
contextuality. In our language their contextual regime is the operational shadow of
proximity to $\partial\mathcal{D}$. The osmotic divergence localizes in space and time
what the contextuality test certifies statistically, two views of the same nodal locus.

\subsection{A two-node test of additivity, and what isolates the two-state structure}
\label{sec:additivity}
Everything proposed so far can be read off a single field. Equation~\eqref{eq:kxprediction}
is fixed by the recorded intensity, and in the Kocsis geometry the post-selection is on
position, so $\phi^*\psi\propto\psi$ and the two-state structure is not exercised at all.
A referee from singular optics is entitled to call that a single-field measurement with
extra steps, and the objection is correct as far as it goes. One prediction of
Theorem~\ref{thm:main} escapes it, and it is the additivity of Step~4 of that proof.

If $\psi$ carries a zero of order $k_\psi$ and $\phi$ a zero of order $k_\phi$ at the same
point, the signed amplitude $\phi^*\psi$ carries one of order $k=k_\psi+k_\phi$ and the pole
coefficient is the sum. At finite visibility each zero is displaced off the real axis by
its own $\epsilon$ (Proposition~\ref{prop:genericity}), and the osmotic channel is then the
sum of two dispersive profiles,
\begin{equation}
\Ima(k_{x,w})=-\frac{x-x_\psi}{(x-x_\psi)^2+\epsilon_\psi^2}
-\frac{x-x_\phi}{(x-x_\phi)^2+\epsilon_\phi^2},
\label{eq:twonode}
\end{equation}
which follows from $\Ima(k_{x,w})=-\tfrac12\partial_x\ln(\rho_\psi\rho_\phi)$ and reduces to
twice \eqref{eq:kxprediction} when $x_\psi=x_\phi$ and $\epsilon_\psi=\epsilon_\phi$
(verified symbolically). No single-field measurement produces \eqref{eq:twonode}, because
it requires two independently prepared amplitudes. This is the one prediction in the paper
that Refs.~\cite{berry2009,bliokh2013,flackhiley2018} do not already contain.

The implementation keeps the Kocsis bench and changes the post-selection. The preparation
supplies $k_\psi=1$ for free: every dark fringe of the two-slit pattern is a first-order
zero of $\psi$. What is needed is a $\phi$ with its own zero at a transverse position the
experimenter controls, which rules out the position post-selection used in that run and
calls for a mode-projective one. Projecting onto a first-order Hermite--Gaussian
$\mathrm{HG}_1$, whose single node is placed by translating the projecting element, does
it: a $\pi$-step phase plate or a phase pattern on a spatial light modulator in the imaging
plane, followed by coupling into a single-mode fiber, implements
$\phi=\mathrm{HG}_1$ with $k_\phi=1$. Both $\psi$ and $\phi$ are then in the real-envelope
class of Remark~\ref{rem:scope}(ii), so the signed factorization
\eqref{eq:signedpolar} holds and Theorem~\ref{thm:main} applies with $k=2$. Neither the
calcite, the weak coupling, nor the linear-basis readout of \eqref{eq:qubitpointer} changes.

Sliding the $\mathrm{HG}_1$ node across the chosen dark fringe turns
\eqref{eq:twonode} into a one-parameter family, and the family is the measurement. At
separations $|x_\psi-x_\phi|\gg\epsilon$ the profile shows two resolved dispersive features
of coefficient $1$ each; as the separation shrinks below $2\epsilon$ they merge into a
single feature of coefficient $2$, of twice the peak height of \eqref{eq:kxprediction} and
the same width. For the geometry of Sec.~\ref{sec:experiment}, $2\epsilon\approx169~\mu$m
at $V=0.90$ and $51~\mu$m at $V=0.99$, i.e.\ $2$--$7$ CCD pixels, so both the resolved and
the merged regimes lie within one scan of the translation stage. What is being measured is
a coefficient going from $1{+}1$ to $2$ as a function of a mechanical displacement, which
is a stronger statement than a single peak height because the control and the signal come
from the same run and share every calibration.

The magnitudes follow from doubling the numbers already quoted. The peak
$|\Ima(k_{x,w})|/k_0$ at the merged node runs from $1.8\times10^{-3}$ at $V=0.90$ to
$5.9\times10^{-3}$ at $V=0.99$, and with $\zeta=373.5$ the pointer rotation
$\varphi_I=\zeta\,\Ima(k_{x,w})/k_0$ runs from $0.66$ to $2.2$~rad. The upper end is not
usable at fixed coupling: $\langle\sigma_z\rangle=\tanh(2.2)=0.976$ is deep in saturation,
and the merged and resolved regimes would be indistinguishable at the peak. The test is
therefore best run at moderate visibility, $V\simeq0.90$--$0.95$, where
$\tanh\varphi_I$ moves from $0.32$ to $0.58$ ($V=0.90$) and from $0.45$ to $0.75$
($V=0.95$) as the two nodes merge, a change far above the demonstrated pointer sensitivity
and still within the monotonic part of \eqref{eq:qubitpointer}. Alternatively $\zeta$ is
reduced to restore linearity at high visibility, at the usual cost in exposure.

Two costs are specific to this configuration and neither is hidden. The first is the
post-selection probability. Position post-selection retains a fixed fraction of the
incident flux at every plane, whereas mode projection onto $\mathrm{HG}_1$ retains
$|\langle\phi|\psi\rangle|^2$, which for a two-slit $\psi$ and an aligned $\mathrm{HG}_1$
is of order a few percent and falls as the node alignment improves. The second is the
rare-event scaling of Sec.~\ref{sec:experiment}: at a merged node $k=2$, so the incident
count needed to resolve the profile down to $\delta x$ grows as
$\delta x^{-(2k+1)}=\delta x^{-5}$ rather than $\delta x^{-3}$. Together these put the
additivity run at one to two orders of magnitude more exposure than the single-node
prediction of the preceding subsection, hours rather than minutes per plane, and the
experiment is worth that only because a positive result cannot be reproduced by any
single-field measurement.

Two controls come with the same apparatus. Removing the two-slit mask, so that $\psi$ is a
nodeless Gaussian and only $\phi$ carries a zero, must return exactly the $k=1$ profile of
\eqref{eq:kxprediction} with the pole located at the $\mathrm{HG}_1$ node and
$\Rea(k_{x,w})$ bounded through it. That is Proposition~\ref{prop:exchange}, and it is
itself a two-state statement: the osmotic divergence is indifferent to which of the two
states supplies the node, so a zero belonging to the post-selection alone produces a pole
in a channel of the pre-selected beam that has no zero anywhere. The second control is the
current channel. Across the merged node $\Rea(k_{x,w})$ must remain bounded and, in the
real-envelope class, spatially uniform up to the smooth contribution of the nonvanishing
factors, so the same scan that doubles the osmotic peak must leave the trajectory slopes
Kocsis reconstructed qualitatively unchanged. A divergence appearing in both channels
would indicate that the projecting element has introduced a spatially varying phase and
taken $\phi$ out of the real-envelope class, which is a diagnosable failure of the
preparation rather than of Theorem~\ref{thm:main}.

Nothing in this subsection tests Theorem~\ref{thm:noflux}, which is falsified or supported
by the two-image construction of Sec.~\ref{sec:experiment} instead. The three proposals are
independent and answer different questions: the zero-cost discriminator asks which
mechanism limits the recorded fringe contrast, the conjugate-channel readout asks whether
the osmotic profile is where the intensity says it is, and the additivity scan asks whether
the object carrying the obstruction is the two-state amplitude rather than either state
alone.

\subsection{The osmotic field as a spatially resolved contextuality witness}
This correspondence can be made quantitative. To keep every probability dimensionless,
post-select on a finite position bin. Let
$\Pi_x^\Delta=\int_{x-\Delta/2}^{x+\Delta/2}|y\rangle\langle y|\,dy$ be the (idempotent)
window projector of width $\Delta$, small on the scale of variation of the fields, with
post-selection \emph{probability}
$p_F(x)=\Tr(\Pi_x^\Delta\hat\rho)\simeq\rho(x)\,\Delta\in[0,1]$, where $\hat\rho$
is the state operator and $\rho(x)$ its position density.
The numerator of the momentum weak value
post-selected on this bin, $\Tr(\Pi_x^\Delta\hat p\,\hat\rho)$, is the first moment of the
Kirkwood--Dirac quasiprobability~\cite{kirkwood1933,dirac1945}, whose negativity is a
metrological resource in post-selected estimation~\cite{arvidsson2020} and is reviewed
in~\cite{arvidssonshukur2024}. Its imaginary part is,
up to the post-selection probability $p_F(x)$, exactly the
osmotic part $\Ima(p_w)$ of the weak momentum (i.e.\ $-m$ times the osmotic velocity $u$):
\begin{equation}
\begin{aligned}
\Ima\,\Tr(\Pi_x^\Delta\hat p\,\hat\rho)&=p_F(x)\,\Ima(p_w)\\
&=-\tfrac{\hbar}{2}\,\partial_x\rho(x)\,\Delta\;\propto\;\partial_x\rho(x),
\end{aligned}
\label{eq:KDosmotic}
\end{equation}
to leading order in $\Delta$. Both sides are $O(\Delta)$, so the normalized field
$\Ima(p_w)$, the object Theorem~\ref{thm:main} is about, is $\Delta$-independent,
and the $\Delta\to0$ limit recovers the density-level statement
$\Ima\,\Tr(\Pi_x\hat p\,\hat\rho)/\rho(x)=\Ima(p_w)$ per unit length.

Kunjwal, Lostaglio, and Pusey~\cite{kunjwal2019} turn a nonzero \emph{imaginary} weak
value into a violation of an explicit noncontextuality inequality. Writing $p_-$ for the
joint probability of a negative pointer-momentum reading and a successful post-selection
on the same bin $\Pi_x^\Delta$ (so that $p_-$, $p_F$, and $p_d$ are all dimensionless
probabilities at finite $\Delta$),
$s$ for the pointer width, and $p_d=o(1/s)$ for the measurement disturbance, their bound
and its quantum value are
\begin{equation}
\begin{aligned}
p_-&\;\le\;\underbrace{\tfrac12 p_F+(1-p_F)\,p_d}_{\text{noncontextual}},\\
p_-^{\mathrm{QM}}&=\tfrac12 p_F-\frac{\Ima\langle\Pi_x^\Delta\mathcal{E}\rangle_{\hat\rho}}{\sqrt\pi\,s}+o(1/s),
\end{aligned}
\label{eq:klp}
\end{equation}
so contextuality is certified wherever this imaginary Kirkwood--Dirac quasiprobability is
negative enough to beat the disturbance floor:
\begin{equation}
-\,\Ima\langle\Pi_x^\Delta\mathcal{E}\rangle_{\hat\rho}
\;>\;\sqrt\pi\,s\,(1-p_F)\,p_d,
\label{eq:witness}
\end{equation}
where $\mathcal{E}$ is the weakly measured (projector) effect and
$\langle\Pi_x^\Delta\mathcal{E}\rangle_{\hat\rho}=\Tr(\Pi_x^\Delta\mathcal{E}\hat\rho)$ the corresponding
Kirkwood--Dirac quasiprobability. The $\tfrac12 p_F$ of \eqref{eq:klp} has cancelled
between the noncontextual bound and the quantum value, leaving exactly the KLP
condition~\cite{kunjwal2019}. This imaginary quasiprobability carries the osmotic profile
of \eqref{eq:KDosmotic},
$-\Ima\langle\Pi_x^\Delta\mathcal{E}\rangle_{\hat\rho}\propto p_F(x)\,\Ima(p_w)(x)\propto\partial_x\rho(x)$,
with the constant fixed as follows.

Equation~\eqref{eq:KDosmotic} is a statement about the
momentum $\hat p$, whose first moment is dimensionful, whereas the KLP inequality is stated
for a dimensionless effect $0\le\mathcal{E}\le\openone$. We therefore take the effect to be
the band-limited momentum effect
\begin{equation}
\mathcal{E}=\tfrac12\big(\openone+\hat p/p_{\max}\big)\,\Pi_{[-p_{\max},p_{\max}]},
\label{eq:effect}
\end{equation}
with $p_{\max}$ any bound on the support of the transverse momentum distribution (in the
paraxial realization $p_{\max}=\hbar k_0\theta_{\max}$, set by the numerical aperture).
Since $\hat p$ and $\Pi_{[-p_{\max},p_{\max}]}$ are both diagonal in the momentum basis
they commute, so $\mathcal{E}$ is self-adjoint with spectrum $\{0\}\cup[0,1]$ and is a
legitimate effect, $0\le\mathcal{E}\le\openone$. It is not a projector, and two remarks
are needed on applying~\cite{kunjwal2019} to it. First, their quantum-side identity uses
idempotence nowhere. The momentum-basis Kraus operators are proportional to unitaries for
any self-adjoint $\mathcal{E}$, so the pointer-momentum marginal is trivial and the
first-order expansion reproduces the second line of \eqref{eq:klp} verbatim. Second, the
\emph{band limiting} is what keeps the disturbance floor small. The achievable $p_d$
degrades with the spread of the spectrum and remains $o(1/s)$ precisely because
$\mathrm{spec}(\mathcal{E})\subset[0,1]$ is bounded. An unbounded $\hat p$ would not do,
which is why \eqref{eq:effect} band-limits instead of merely rescaling. Since
$\Ima\langle\Pi_x^\Delta\openone\rangle=0$ identically, and provided the state's momentum
distribution is supported in the band,
\begin{equation}
\Ima\langle\Pi_x^\Delta\mathcal{E}\rangle_{\hat\rho}
=\frac{\Ima\,\Tr(\Pi_x^\Delta\hat p\,\hat\rho)}{2p_{\max}}
=-\frac{\hbar\,\Delta}{4p_{\max}}\,\partial_x\rho(x),
\label{eq:effectKD}
\end{equation}
the first equality being exact and the second inheriting, without further approximation,
the $O(\Delta)$ of \eqref{eq:KDosmotic}. The proportionality constant is thus $\hbar\Delta/4p_{\max}$, and the
$\partial_x\rho$ profile is inherited and not assumed. The sign is fixed on either side of the fringe by taking $\mathcal{E}$ or
$\openone-\mathcal{E}$, the same without-loss-of-generality choice as
in~\cite{kunjwal2019}, as the dispersive profile of $\Ima(p_w)$ demands.

The photonic realization needs the imaginary-part inequality for a \emph{qubit} pointer,
which is the Kocsis polarization pointer, rather than for the continuous Gaussian pointer
of~\cite{kunjwal2019}. That version is proved as Proposition~\ref{prop:qubitklp} in
Appendix~\ref{app:qubitklp}. Writing $p_-$ for the joint probability that the pointer is
found in $|V\rangle$ and the post-selection succeeds, and $\theta$ for the dimensionless
coupling, the two statements it supplies are
\begin{equation}
p_-=\tfrac12 p_F-\theta\,\Ima\langle\Pi_x^\Delta\mathcal{E}\rangle_{\hat\rho}+O(\theta^2),
\qquad
p_-\le\tfrac12 p_F+\tfrac12 p_d ,
\label{eq:qubitpair}
\end{equation}
the first exact quantum mechanics and the second the noncontextual bound. The qubit version
is sharper than the continuous one. The witness is first order in $\theta$ while the
measurement disturbance is second order, and no $o(1/s)$ pointer-width remainder survives. Through
\eqref{eq:witness} the osmotic field $\Ima(p_w)(x)$ is a spatially
resolved image of the KLP witness. Its sign change and $1/(x-\xnode)$ divergence trace
where the per-post-selected-event nonclassicality is largest.

The operational witness does not itself diverge at the node. The object that grows without bound and the object an experiment reports are two
distinct quantities. The per-post-selected-event
nonclassicality is the \emph{normalized} imaginary weak value $\Ima(p_w)$, and it is this
that diverges as $1/(x-\xnode)$. The quantity that violates the KLP noncontextuality
inequality, by contrast, is the \emph{unnormalized} Kirkwood--Dirac quasiprobability
$p_F\,\Ima(p_w)$, weighted by the post-selection probability an ensemble actually realizes.
At an order-$k$ node $p_F\sim(x-\xnode)^{2k}$ while $\Ima(p_w)\sim1/(x-\xnode)$, so this
ensemble witness scales as $(x-\xnode)^{2k-1}\to0$ (equivalently, it is
$\propto\partial_x\rho$, which vanishes at any node). The two statements are consistent and
answer different operational questions. \emph{Per retained event}, contextuality is
maximal exactly at the singularity. \emph{Per incident trial} it is suppressed there,
because reaching the node is a vanishingly rare post-selection. This is the same rare-event ledger as the
resource discussion above, now with a sign. The factor $p_F$ that suppresses the ensemble witness is
identically the factor that makes node-neighborhood post-selection expensive
[$N_{\text{inc}}\sim(\delta x)^{-(2k+1)}$]. The osmotic singularity therefore marks the
locus of maximal per-event contextuality while leaving the ensemble witness bounded. Which
of the two an experiment sees is fixed by whether it normalizes to retained events or to
incident trials. This turns Remark~\ref{rem:contextuality} from a structural
correspondence into a quantitative identity, for the explicit effect \eqref{eq:effect},
with the constant given in \eqref{eq:effectKD}, between the imaginary weak-value field
and the Kirkwood--Dirac contextuality witness.

\begin{remark}[The same ledger in an open system]
\label{rem:ledger}
The pattern is not special to the nodal geometry. In the past-quantum-state description of
a dephasing qubit~\cite{gammelmark2013}, with forward and backward coherence factors
$a=e^{-\Gamma(t)}$ and $b=e^{-\Gamma(T-t)}$, the conditioned coherence is $(a+b)/(1+ab)$.
It equals unity at both selected times and exceeds the unconditioned $e^{-\Gamma(t)}$ by
$2\sinh\Gamma t/(e^{\Gamma T}+1)$, hence strictly for $t>0$, so conditioning holds the
coherence up instead of freezing it. The price is
the post-selection success probability
$P_{\text{ps}}=\tfrac12(1+e^{-\Gamma(T)}\cos\omega_0T)$. Retaining $N$ shots costs
$\sim N/P_{\text{ps}}$ attempts, and for $n$ independently post-selected qubits the
overhead grows exponentially in $n$. The conditional gain and the trial cost cancel, so the
effect is not retroactive denoising. The osmotic witness above is subject to the same
accounting.
\end{remark}

Two-time descriptions are already realized experimentally, in cavity-QED photon-number
past quantum states~\cite{rybarczyk2015}. The ledger of Remark~\ref{rem:ledger} applies there in the same form.

\section{Scope and limitations}
\label{sec:caveats}

Nelson diffusions and textbook quantum mechanics agree on every single-time position
distribution, the process marginal being $|\psi(x,t)|^2$ by construction. They
disagree on \emph{multi-time} correlations~\cite{grabert1979,nelson2005}. The two-time
joint distribution of the Nelson trajectory does not, in general, match the quantum
sequential-measurement statistics, because the smooth trajectory ensemble carries no
measurement back-action while the quantum two-time correlator does. This discrepancy is
only partially reconciled through effective collapse~\cite{blanchard1986,bacciagaluppi2022}.
Every claim about \emph{quantum statistics} in this paper is a single-time statement.
The fields $\Rea(p_w)$ and $\Ima(p_w)$, together with the drift decomposition at a fixed
$t$, are therefore untouched by the gap. Statements at the level of paths, meaning the
inaccessibility of $\partial\mathcal{D}$ in Proposition~\ref{prop:inaccessible} and the
sampled ensemble of Fig.~\ref{fig:traj}, are genuinely multi-time, but they concern
the mathematical diffusion itself and are not asserted to reproduce quantum
sequential-measurement statistics. The gap does, however, bound the interpretation of those trajectories, whose individual
paths should not be read as physical two-time histories. Their marginal
is the constructed $\bar\rho$ of Sec.~\ref{sec:numerics} rather than the conditioned
product density (Proposition~\ref{prop:current}), and they serve only as a device for exhibiting the current
velocity and its osmotic partner.

We also assume the Wallstrom single-valuedness condition on the
phase~\cite{wallstrom1994}, which holds trivially for the Gaussian and
Hermite--Gaussian states used throughout but is not automatic for a general
wavefunction.

The confinement to $\mathcal{D}$ is not a restriction at all, though it is easily read as
one. Nothing here extends the positive conditioning across $\partial\mathcal{D}$ because
nothing can, that failure being the content of Corollary~\ref{cor:conditioning} and not a
gap in its proof.

One point about the fields themselves is easy to miss on a first reading.
$\Rea(p_w)$ is finite straight through
$\partial\mathcal{D}$. Only $\Ima(p_w)$ diverges there. A figure or an argument built
entirely from $\Rea(p_w)$, from a ``bridge drift,'' for instance, will not show the
boundary at all, and that omission is exactly how the caustic of
Sec.~\ref{sec:nodecaustic} can be mistaken for it.

\section{Conclusion}

The conditioned-diffusion representation of the weak-value flow field fails in two ways
that should be kept apart. Identifying the weak momentum with a bridge drift fails even
where a conditioned diffusion exists, the relation there being exact and a sign inversion.
For a free particle post-selected in position the current velocity is exactly minus the
drift of the bridge to the target. Where the two-state amplitude has a moving zero across
which probability flows, no diffusion of constant diffusion coefficient can both carry the
conditioned density and avoid the nodal curve (Theorem~\ref{thm:noflux}); the drift such a
diffusion would need points into the curve, which then lies at finite scale distance and is
reached in finite time with positive probability. The second
statement is the more general. It needs no state class, only a node across which mass is
transported. What does have a boundary is the positive conditioning itself. It sits on the union of the nodal sets of the
pre- and post-selected states, a curve $\xnode(t)$ traced out in space and time rather
than the caustic time-slice one might expect from the delta-function limit. And it shows
up only in the osmotic momentum $\Ima(p_w)$, which diverges as $1/(x-\xnode)$ while the
current velocity $\Rea(p_w)$ passes straight through. That asymmetry is the main point. It
says precisely where a classical-conditioning account of time-symmetric quantum phenomena
has to break down, and it identifies the imaginary part of the weak value as the part that
carries the break, consistent with the role imaginary weak values already play in proofs
of contextuality.

The current-velocity field has been
measured directly, in two-slit weak-measurement experiments, and, as our result
requires, it does not diverge at the interference nodes (Sec.~\ref{sec:experiment}).
The osmotic field at those same nodes has not been measured, and our result says
exactly what should be seen if it is measured, a divergence with the sign-changing
$1/(x-\xnode)$ profile of Eq.~\eqref{eq:imdiv}, read out through the momentum-pointer
channel of the same class of experiment.

A complementary statement appeared while this work was being completed: for free
evolution with fixed endpoint moduli, the presence of nodes demotes the quantum
evolution from a global to a local minimum of a quadratic transport
cost~\cite{morato2026}. Nodes are unstable under free dynamics there; they are an
impassable boundary of a two-time conditioning here, an optimality statement beside an
existence one.

Several questions remain open. The two-endpoint bridge that pins both the pre- and
post-selected states follows the classical connecting trajectory, not the weak value
(Appendix~\ref{app:bernstein}). Identifying the stochastic object whose drift is
$\Rea(p_w)$ for a moving, phase-carrying post-selection, if one exists beyond the
phase-gradient and Wiseman readings, is left open. In more than one dimension the nodal
set becomes a codimension-two surface rather than a point, so $\partial\mathcal{D}$ is a
higher-dimensional locus. The Madelung split and Theorem~\ref{thm:main} apply verbatim
(Remark~\ref{rem:kinematic}), but the geometry of the osmotic divergence and the
node-inaccessibility argument deserve separate treatment there
(Remark~\ref{rem:scope}).

The sharpest open question is whether
the per-event contextuality that the osmotic field maps (Sec.~\ref{sec:experiment}) can be
turned into an operational witness that survives the post-selection cost rather than being
cancelled by it, the one place where the ledger that recurs throughout this paper might,
or might not, be beaten. The dependence of the divergence strength on the nodal order $k$
(Theorem~\ref{thm:main}) is an experimental discriminant, and the route to it is additivity
rather than a single higher-order mode. No Hermite--Gaussian mode supplies $k\ge2$ on its
own, the zeros of $H_n$ being simple (Remark~\ref{rem:scope}), but a simple zero in the
preparation and a simple zero in the post-selection placed at one transverse position give
$k=k_\psi+k_\phi=2$ in the amplitude, with the pole coefficient doubled and the pole itself
still simple in $x-\xnode$. That is the scan of Sec.~\ref{sec:additivity}, and it is the
only measurement proposed here whose outcome no single-field experiment can reproduce.

One ledger runs through everything above. Each two-time construction in this paper buys a
conditional advantage per retained event, contextuality maximal at the node and conditioned
coherence pinned at the endpoints, and pays for it in incident trials. Nothing is
manufactured on average. Wherever a claim in this subject looks like a
free lunch, the first question to ask is which of the two normalizations, per event or
per trial, is being quoted.

\begin{acknowledgments}
We thank Dr.\ Matisse Wei-Yuan Tu for a helpful discussion.
\end{acknowledgments}

\paragraph*{Code and data availability.}
A single self-contained \texttt{sympy}/\texttt{numpy} script reproduces every figure,
every table, and every verification claim in this paper. It is openly available under a
permanent identifier, Ref.~\cite{zenodo2026}. Each claim marked
``(verified symbolically)'' or ``(verified numerically)'' in the text corresponds to a
tagged routine in that script, as does each figure and each of
Tables~\ref{tab:scan} and~\ref{tab:threebox}, so that text and code can be mapped onto each
other line by line. The mapping is listed in the repository's \texttt{README}.

\paragraph*{Author contributions.}
C.-F.K. conceived the study, developed the theory and the symbolic and numerical
verifications, and wrote the manuscript. K.-W.W. proposed the reciprocal-process
(Bernstein-bridge) framing of the conditioning problem. Both authors discussed the
results and reviewed the manuscript.

\appendix

\section{Derivation of the two-field split}
\label{app:madelung}
Everything in the paper rests on the split \eqref{eq:re}--\eqref{eq:im} and on the
operational identity quoted below it, so both are shown here in full, together with the
conventions they depend on and the two extensions used later.

\emph{Conventions.} The pre-selected state is written
$\psi=\sqrt{\rho_\psi}\,e^{iS_\psi/\hbar}$ and the post-selected one
$\phi=\sqrt{\rho_\phi}\,e^{-iS_\phi/\hbar}$, the opposite sign in the second exponent
recording that $\phi$ is propagated backward from the post-selection time. Explicitly, in
the two-state-vector convention $\phi(x,t)=\langle x|U(t,T)|\phi_f\rangle$, so that
$\phi^*(x,t)=\langle\phi_f|U(T,t)|x\rangle$. For a delta post-selection at $x_f$ and time $T$,
$|\phi_f\rangle=|x_f\rangle$, this makes $\phi^*$ the forward propagator $K(x_f,T|x,t)$ and
$S_\phi$ the classical two-point action of Appendix~\ref{app:mehler}. This has two
consequences used throughout. First, $v_\phi:=\partial_xS_\phi/m$ is \emph{minus} the velocity of the
post-selected center, since $S_\phi$ is the action accumulated from $t$ forward to $T$
rather than backward from $T$. This is the sign that reappears in \eqref{eq:xdotcl} and in
Appendix~\ref{app:bernstein}, and it is a convention about labelling, not about physics,
in the sense of Proposition~\ref{prop:exchange}. Second, $\rho_\psi=|\psi|^2$ and
$\rho_\phi=|\phi|^2$ are unnormalized in general, which is harmless because only
$\nabla\ln$ of them appears.

\emph{The split.} With those conventions,
\begin{equation}
\phi^*\psi=\sqrt{\rho_\psi\rho_\phi}\,\exp\!\Big[\tfrac{i}{\hbar}(S_\psi+S_\phi)\Big].
\end{equation}
Taking the logarithm splits it into a real and an imaginary part by inspection,
\begin{equation}
\ln(\phi^*\psi)=\tfrac12\ln(\rho_\psi\rho_\phi)+\tfrac{i}{\hbar}(S_\psi+S_\phi),
\end{equation}
and the gradient, being linear, acts on each part separately:
\begin{equation}
\nabla\ln(\phi^*\psi)=\tfrac12\nabla\ln(\rho_\psi\rho_\phi)+\tfrac{i}{\hbar}\nabla(S_\psi+S_\phi).
\label{eq:gradlog}
\end{equation}
Multiplying \eqref{eq:gradlog} by $-i\hbar$ gives $p_w=-i\hbar\nabla\ln(\phi^*\psi)$
directly as a sum of a real and an imaginary term,
\begin{equation}
p_w=\nabla(S_\psi+S_\phi)\;-\;\tfrac{i\hbar}{2}\nabla\ln(\rho_\psi\rho_\phi),
\end{equation}
and reading off the real and imaginary parts term by term reproduces
\eqref{eq:re}--\eqref{eq:im} exactly, with no approximation at any step. The split is
algebraic rather than perturbative, holds for any $\psi,\phi$ that are $C^1$ and nonzero
at the point in question, and holds in any number of dimensions, $\nabla$ being the full
gradient. Only the \emph{asymmetry} of Theorem~\ref{thm:main} is one-dimensional
(Remarks~\ref{rem:scope} and~\ref{rem:kinematic}).

\emph{Single-valuedness.} The complex logarithm is multivalued, so \eqref{eq:gradlog}
needs one remark. On any simply connected subset of $\mathcal{D}=\{\phi^*\psi\neq0\}$ a
branch of $\ln(\phi^*\psi)$ exists and is $C^1$, and two branches differ by a constant
$2\pi i n$, whose gradient vanishes. Hence $\nabla\ln(\phi^*\psi)$, and therefore $p_w$,
is single-valued on all of $\mathcal{D}$ even where no global branch of the logarithm is.
What is \emph{not} guaranteed is that the phase $\Theta$ be $C^1$ with bounded gradient
across a zero, which is a separate condition and is exactly the hypothesis
\eqref{eq:signedpolar} of Theorem~\ref{thm:main}. It fails when $\phi^*\psi$ winds, as at
an optical vortex [Remark~\ref{rem:scope}(i)].

\emph{The operational identity.} The field $p_w=-i\hbar\,\partial_x\ln(\phi^*\psi)$ is the
difference of two differently ordered Kirkwood--Dirac weak values, as asserted in
Sec.~\ref{sec:madelung}. With $\Pi_x=|x\rangle\langle x|$, the first ordering gives
$\langle\phi|\Pi_x\hat p|\psi\rangle=\phi^*(x)\,(-i\hbar\,\partial_x\psi)$. For the second,
$\langle\phi|\hat p|x\rangle=\overline{\langle x|\hat p|\phi\rangle}
=\overline{-i\hbar\,\partial_x\phi}=+i\hbar\,\partial_x\phi^*$, so
$\langle\phi|\hat p\,\Pi_x|\psi\rangle=i\hbar\,(\partial_x\phi^*)\,\psi(x)$. Subtracting and
using the product rule,
\begin{equation}
\langle\phi|\Pi_x\hat p|\psi\rangle-\langle\phi|\hat p\,\Pi_x|\psi\rangle
=-i\hbar\,\partial_x(\phi^*\psi),
\end{equation}
and dividing by $\langle\phi|\Pi_x|\psi\rangle=\phi^*\psi$ gives
$-i\hbar\,\partial_x\ln(\phi^*\psi)$, which is the identity quoted. Neither ordering alone
returns $p_w$. The first gives $-i\hbar\,\partial_x\ln\psi$, the standard local weak
momentum, and the second $+i\hbar\,\partial_x\ln\phi^*$. They coincide with $p_w$ only when
the other factor is constant in $x$, which for the first is the same-plane position
post-selection used by Kocsis \emph{et al.}~\cite{kocsis2011}, where $v_\phi=0$. That case
is distinct from the delta post-selection at $(x_f,T)$ above, for which $v_\phi\neq0$ and
Eq.~\eqref{eq:xdotcl} applies.

\emph{The mixed-state osmotic identity.} Equation~\eqref{eq:KDosmotic} and
Proposition~\ref{prop:qubitklp} need the osmotic half of the split for a general density
matrix rather than a pure pair. Writing $\hat\rho(y,x)=\langle y|\hat\rho|x\rangle$ and
$\rho(x)=\hat\rho(x,x)$,
\begin{equation}
\langle x|\hat p\,\hat\rho|x\rangle=-i\hbar\,\partial_y\hat\rho(y,x)\big|_{y=x},
\end{equation}
while hermiticity $\hat\rho(x,y)=\overline{\hat\rho(y,x)}$ gives
$\partial_x\rho(x)=2\,\Rea\,\partial_y\hat\rho(y,x)|_{y=x}$. Therefore
\begin{equation}
\Ima\,\langle x|\hat p\,\hat\rho|x\rangle
=-\hbar\,\Rea\,\partial_y\hat\rho(y,x)\big|_{y=x}
=-\tfrac{\hbar}{2}\,\partial_x\rho(x),
\label{eq:mixedosmotic}
\end{equation}
independent of any purity assumption. For a pure two-state pair, dividing by
$\rho_\psi\rho_\phi$ recovers \eqref{eq:im}. This is why the osmotic prediction of
Sec.~\ref{sec:experiment} is a statement about the recorded intensity alone and survives
partial coherence.

\emph{Symbolic check.} We have verified \eqref{eq:gradlog} with \texttt{sympy}, checking
that the real and imaginary parts of the direct complex-logarithm derivative agree with the
term-by-term reading above once $S_\psi,S_\phi,\rho_\psi,\rho_\phi$ are declared real and
positive. Without that declaration a computer-algebra system does not know that
$\ln|\cdot|$ is real, which is where an unevaluated symbolic
expression can silently disagree with the intended one.

\section{Two discrete analogues, the three-box and Hardy paradoxes}
\label{app:discrete}
Two standard discrete pre- and post-selection paradoxes exhibit, in zero dimensions, a
feature that drives the nodal case. The continuous node $\phi^*\psi=0$ is replaced by a
conditioning weight that changes sign or is negative, and an anomalous, out-of-range weak
value appears exactly where that weight fails to be a positive overlap. Neither is an
instance of Theorem~\ref{thm:main}, for the reason given at the end of this appendix, and
we use them as illustrations of the mechanism and not as corollaries.

\emph{A continuous embedding of the three boxes.} Take three Gaussian wells
$b_i(x)=e^{-(x-r_i)^2/2\sigma^2}$ centered at $r_1,r_2,r_3=-d,0,d$, and set
$\psi\propto b_1+b_2+b_3$ and $\phi\propto b_1+b_2-b_3$, the minus sign carrying the
post-selection that makes the third box the disallowed one. Weak projector values are
computed on the bins $W_1=(-\infty,-d/2)$, $W_2=(-d/2,d/2)$, $W_3=(d/2,\infty)$ as
\begin{equation}
(P_i)_w=\frac{\int_{W_i}\phi^*\psi\,dx}{\int_{-\infty}^{\infty}\phi^*\psi\,dx},
\label{eq:threeboxdef}
\end{equation}
so that they sum to exactly $1$, as any complete set of weak projector values must. The
integrand $\phi^*\psi=b_1^2+b_2^2-b_3^2+2b_1b_2$ has no cross term with $b_3$, the wells
$1$ and $3$ being separated by $2d$, and every remaining integral is a Gaussian error
function. The result is carried by the two parameters
\begin{equation}
\epsilon=e^{-d^2/4\sigma^2},\qquad
\delta=\tfrac12\,\mathrm{erfc}(d/2\sigma),
\end{equation}
the adjacent-well overlap and the tail of a single well past the nearest bin boundary.
Both are local to this appendix. The symbol $\epsilon$ here is not the off-axis
regularization of Proposition~\ref{prop:genericity} that enters the visibility relation
\eqref{eq:kxprediction}, and $\delta$ is neither the Dirac $\delta$ nor the node distance
$\delta x$ used elsewhere.

\emph{The exact ratios.} Carrying out the integrals in \eqref{eq:threeboxdef} gives
\begin{equation}
(P_1)_w\simeq\frac{1+\epsilon}{1+2\epsilon},\;
(P_2)_w\simeq\frac{1-2\delta+\epsilon}{1+2\epsilon},\;
(P_3)_w\simeq\frac{2\delta-1}{1+2\epsilon}.
\label{eq:threeboxexact}
\end{equation}
The one approximation made is to discard the far tail of the shared cross term $2b_1b_2$
beyond the opposite bin boundary, of order $\epsilon\,\mathrm{erfc}(d/\sigma)$. Since
$\mathrm{erfc}(d/\sigma)$ falls as $e^{-d^2/\sigma^2}=\epsilon^4$ up to an algebraic
prefactor, that term is $O(\epsilon^5)$, smaller than the terms kept in
\eqref{eq:threeboxexact} by a further factor $\epsilon^4$.

\emph{Numerical check.} Equation~\eqref{eq:threeboxexact} matches the quadrature values of
Table~\ref{tab:threebox} to the displayed precision for every $d/\sigma\geq3$, and to five
significant figures in the underlying data. The interesting case is $d/\sigma=2$, where
$\epsilon\simeq0.4$ is not small and a naive $O(\epsilon)$ truncation would fail. There
the residuals are $-9.9\times10^{-4}$ in $(P_2)_w$ and $+1.0\times10^{-3}$ in $(P_3)_w$,
of opposite sign because the three values sum to $1$ both exactly and in quadrature, so
their residuals must cancel. The leading neglected term evaluates to
$\epsilon\,\mathrm{erfc}(d/\sigma)/(1+2\epsilon)=9.91\times10^{-4}$, accounting for both to
better than $10^{-5}$ in magnitude once included. The quadrature was converged to eight
digits. The quantities
being checked are the exact ratios of \eqref{eq:threeboxexact}, not their small-$\epsilon$
limit. In that limit the formulas reduce to $(P_1)_w\to1-\epsilon$ and
$(P_3)_w\to-1+2\delta+2\epsilon$, which are the estimates one would guess by inspection,
and which are wrong by about $20\%$ for $(P_1)_w$ at $d/\sigma=2$.

\emph{What the embedding shows.} The mechanism is visible directly in
Fig.~\ref{fig:threebox}. The conditioning weight $\phi^*\psi$ crosses zero between the
second and third wells, and that crossing is what forces $(P_3)_w$ negative. The
embedding is an illustration of the sign-change mechanism and of its exact rate of
convergence, not an independent derivation of the discrete three-box result. The negative
$(P_3)_w$ is not merely formal, having been observed in the weak-measurement realization
of the three-box problem by Resch, Lundeen, and Steinberg~\cite{resch2004}.

\begin{figure}[!htbp]
\centering
\includegraphics[width=\columnwidth]{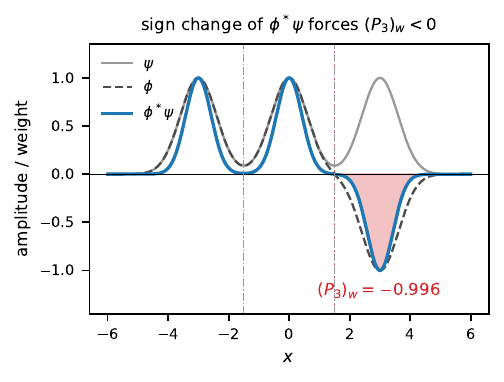}
\caption{Three-box witness. In this continuous embedding the conditioning weight
$\phi^*\psi$ (blue) crosses zero between the second and third wells, and that crossing
forces $(P_3)_w<0$, converging to $-1$ exponentially in the well overlap
[Eq.~\eqref{eq:threeboxexact}]. Plotted for $d/\sigma=5$ ($d=3$, $\sigma=0.6$), where
$(P_3)_w=-0.996$ (Table~\ref{tab:threebox}).}
\label{fig:threebox}
\end{figure}

\emph{Hardy's paradox.} The second example~\cite{hardy1992,aharonov2002} shows the same
feature with no spatial structure at all. With
electron arms $\{O^-,NO^-\}$ and positron arms $\{O^+,NO^+\}$, annihilation in the
overlapping pair $O^-O^+$ leaves the pre-selected state
$\psi\propto|O^-NO^+\rangle+|NO^-O^+\rangle+|NO^-NO^+\rangle$, and post-selection on both
dark ports gives $\phi\propto|O^-O^+\rangle-|O^-NO^+\rangle-|NO^-O^+\rangle+|NO^-NO^+\rangle$.
The conditioning weight is the overlap $\langle\phi|\psi\rangle=-1/(2\sqrt3)$, which is
negative, and that sign alone forces the anomaly. Writing $N_w=\langle\phi|N|\psi\rangle/
\langle\phi|\psi\rangle$ for each occupation, the numerators are read off directly from the
two expansions above, and dividing by the common negative denominator gives
\begin{equation}
N(O^-NO^+)_w=N(NO^-O^+)_w=1,\quad N(O^-O^+)_w=0,
\end{equation}
together with the anomalous joint value
\begin{equation}
N(NO^-NO^+)_w=-1,
\end{equation}
the four summing to $1$ as a complete set must. All four are verified symbolically and are
shown in Fig.~\ref{fig:hardy}. They were observed in weak-measurement realizations of the
optical Hardy's paradox by Lundeen and Steinberg~\cite{lundeen2009} and by Yokota \emph{et
al.}~\cite{yokota2009}.

\emph{Why these are analogues and not instances.} In the genuinely discrete problems the
conditioning weight is negative and not zero. There is therefore no node, nothing
diverges, and Theorem~\ref{thm:main} does not apply. The two examples do share with
the nodal case the weaker and still substantive feature that an out-of-range weak value
appears exactly where $\phi^*\psi$ fails to be a positive overlap, which is the same
failure of positivity that Corollary~\ref{cor:conditioning} identifies as the boundary of
the conditioning. The continuous case adds to this a divergence and a rate. The discrete
cases add the fact that the phenomenon survives with no spatial structure whatsoever.

\begin{figure}[!htbp]
\centering
\includegraphics[width=\columnwidth]{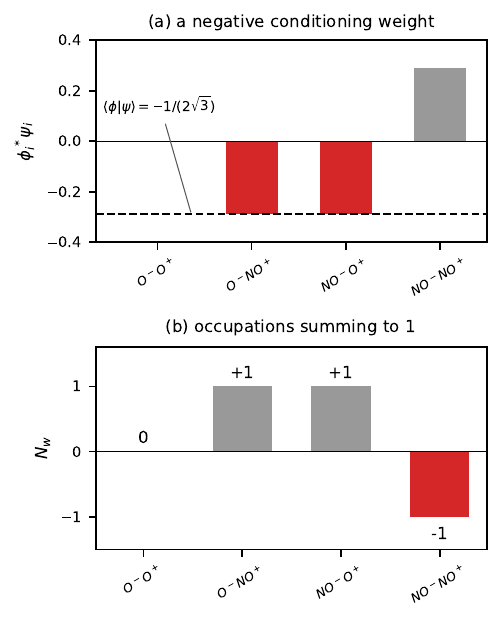}
\caption{Hardy's paradox as a discrete analogue of the sign-change mechanism.
\textbf{(a)} The contributions $\phi^*_i\psi_i$ of each basis state to the conditioning
weight sum to a \emph{negative} overlap $\langle\phi|\psi\rangle=-1/(2\sqrt{3})$ (dashed),
the discrete analogue of the two-state amplitude $\phi^*\psi$ passing through zero at a
node. \textbf{(b)} The occupation weak values that result, $N(O^-NO^+)_w=N(NO^-O^+)_w=1$,
$N(O^-O^+)_w=0$, and the anomalous $N(NO^-NO^+)_w=-1$ (red), forced by the negative
conditioning weight, with the four summing to $1$.}
\label{fig:hardy}
\end{figure}

\begin{table}[!htb]
\caption{Continuous three-well weak projector values versus separation $d/\sigma$
(numerical quadrature). Convergence to the exact $(1,1,-1)$ is exponential in
$\epsilon=e^{-d^2/4\sigma^2}$.}
\label{tab:threebox}
\begin{ruledtabular}
\begin{tabular}{ccccc}
$d/\sigma$ & $\epsilon$ & $(P_1)_w$ & $(P_2)_w$ & $(P_3)_w$\\
\hline
$2$ & $3.7\times10^{-1}$ & $0.788$ & $0.696$ & $-0.484$\\
$3$ & $1.1\times10^{-1}$ & $0.913$ & $0.885$ & $-0.798$\\
$4$ & $1.8\times10^{-2}$ & $0.982$ & $0.978$ & $-0.960$\\
$5$ & $1.9\times10^{-3}$ & $0.998$ & $0.998$ & $-0.996$\\
$6$ & $1.2\times10^{-4}$ & $1.000$ & $1.000$ & $-1.000$\\
\end{tabular}
\end{ruledtabular}
\end{table}

\section{Reciprocal (Bernstein) process and positivity}
\label{app:bernstein}

\subsection{The general construction}
\emph{The pair of solutions.} The Euclidean two-time construction keeps two strictly
positive functions instead of one density~\cite{zambrini1986,leonard2014}. They solve
dual heat equations,
\begin{equation}
\partial_t\eta=\nu\,\partial_x^2\eta-\tfrac{1}{\hbar}V\eta,\qquad
\partial_t\eta^*=-\nu\,\partial_x^2\eta^*+\tfrac{1}{\hbar}V\eta^*,
\label{eq:dualheat}
\end{equation}
the first running forward in time from the initial data and the second backward from the
final data, so that $\eta^*$ is the space-time harmonic factor. For a free particle the
potential term drops and \eqref{eq:dualheat} reduces to the dual heat equations; for the
oscillator of Sec.~\ref{app:statnode} it is the imaginary-time Schr\"odinger operator
$H=-\nu\hbar\,\partial_x^2+V$, whose eigenfunctions supply the closed form there. The pair
is fixed by the two-endpoint conditions
\begin{equation}
\eta(x,0)\,\eta^*(x,0)=\rho_0(x),\qquad
\eta(x,T)\,\eta^*(x,T)=\rho_T(x),
\label{eq:bridgebc}
\end{equation}
which is the system whose solvability is asserted by Corollary~\ref{cor:conditioning}. That the
endpoint data enter only through the \emph{product} is the structure Schr\"odinger
introduced when he posed the problem~\cite{schrodinger1931,schrodinger1932}, reading the
product of two oppositely propagating positive solutions as the classical counterpart of
$\psi\bar\psi$. The object whose positivity is tracked in this paper, $\phi^*\psi$, is the
real-time version of that product.

\emph{The velocity fields.} The associated process has density
$\rho_{\text{eff}}=\eta\eta^*$ and, in the notation of Sec.~\ref{sec:madelung},
\begin{equation}
v=\nu\,\partial_x\ln\frac{\eta^*}{\eta},\qquad
u=\nu\,\partial_x\ln(\eta\eta^*),
\label{eq:bridgevu}
\end{equation}
so that the forward and backward It\^o drifts collapse to one factor each,
\begin{equation}
b_+=v+u=2\nu\,\partial_x\ln\eta^*,\qquad
b_-=v-u=-2\nu\,\partial_x\ln\eta .
\label{eq:bridgedrifts}
\end{equation}
This is the sense in which the phase information survives the Euclidean construction. It
is carried by the ratio $\eta^*/\eta$, which a single density cannot reproduce.
Equivalently, $\eta^*$ is a positive space-time harmonic function and the process is the
Doob $h$-transform \eqref{eq:htransform} with $h=\eta^*$.

\emph{Where it fails.} On each nodal cell $\eta,\eta^*>0$, the drifts
\eqref{eq:bridgedrifts} are smooth and the diffusion is a finite-energy process in
Carlen's sense~\cite{carlen1984}. Positivity fails exactly when $\eta^*$ acquires a zero,
that is on $\partial\mathcal{D}$, where $\partial_x\ln\eta^*$ and hence $b_+$ diverge. By
Theorem~\ref{thm:main} the divergence is osmotic, the real-time current velocity
$\Rea(p_w)/m$ remaining bounded there. The real-time identity
$\Rea(p_w)/m=v_\psi+v_\phi$ is exact at the level of the phase gradient
(Appendix~\ref{app:madelung}). A closed form for the pair $(\eta,\eta^*)$ realizing it is available only for a stationary
node, and is given in Sec.~\ref{app:statnode} below.

\subsection{Real-time verification for coherent pre- and post-selection}
\emph{The field.} We evaluate the coherent Gaussian class in closed form, in the units
$\hbar=m=\omega=1$ of Sec.~\ref{sec:numerics}. Take coherent pre- and post-selection with
phase-space centers $q_\psi(t)=\alpha\cos t$ and $q_\phi(t)=\beta\cos(T-t)$, together with
their conjugate momenta. The amplitudes are written $\alpha,\beta$ to keep them apart from
the drifts $b,b_\pm$. The real-time weak momentum then evaluates to
\begin{equation}
\Rea(p_w)=-\alpha\sin t-\beta\sin(T-t),
\label{eq:cohRe}
\end{equation}
which is the phase gradient $\partial_x(S_\psi+S_\phi)$ of $\phi^*\psi$ (exact, verified
symbolically), spatially uniform, and equal to $v_\psi+v_\phi=\dot q_\psi-\dot q_\phi$.
Operationally it is the Wiseman current velocity~\cite{wiseman2007} of the
position-post-selected ensemble, reconstructed in Ref.~\cite{kocsis2011}. It should be contrasted with the transport
velocity of the naive product density $\rho_\psi\rho_\phi$, a Gaussian centered at
$\tfrac12(q_\psi+q_\phi)$ and therefore transporting at $\tfrac12(\dot q_\psi+\dot
q_\phi)$. The two differ by $\tfrac12(\dot q_\psi-3\dot q_\phi)$, verified symbolically,
and this difference vanishes only where $\dot q_\psi=3\dot q_\phi$. For $\alpha,\beta>0$
and $0<T<\pi$ the two sides of that condition, $-\alpha\sin t$ and $3\beta\sin(T-t)$, have
opposite signs on $(0,T)$ and exactly one of them vanishes at each endpoint, so the locus
is empty on $[0,T]$. For larger $T$ isolated crossings with both centers moving can occur,
which does not affect the genericity statement. This is the concrete content of
Proposition~\ref{prop:current}.

\emph{Existence without a closed form.} Existence of the positive reciprocal bridge on
each nodal cell, where the endpoint densities are strictly positive, is guaranteed by the
Beurling--Jamison theory of the Schr\"odinger
problem~\cite{beurling1960,jamison1974,jamison1975}. No closed form for that pair is available in general, and the two-sided conditioning
reproducing this field is exhibited explicitly only for a stationary node. For
a moving (coherent) node the two-endpoint-pinned reciprocal bridge follows, once continued
back to real time, the classical trajectory connecting the pre- and post-selected centers,
whose velocity $(-\alpha\cos(T-t)+\beta\cos t)/\sin T$ solves $\ddot x=-\omega^2x$ with
$x(0)=\alpha$, $x(T)=\beta$ and differs from $\Rea(p_w)/m$ (verified symbolically). Before
continuation the Euclidean bridge is the hyperbolic interpolation
$[\alpha\sinh\omega(T-t)+\beta\sinh\omega t]/\sinh\omega T$, consistent with
\eqref{eq:oscbridge}. The trigonometric and hyperbolic forms are the two sides of the
$t\to-it$ continuation of Appendix~\ref{app:beyond}, and neither is $\Rea(p_w)/m$.
There is thus no closed-form moving Bernstein bridge whose drift is the weak value, which
is consistent with the thesis of this paper and with Theorem~\ref{thm:noflux}. The identity
$\Rea(p_w)/m=v_\psi+v_\phi$ does not depend on any such bridge. It rests on the exact phase
gradient of $\phi^*\psi$ and, for position post-selection ($v_\phi=0$, the Kocsis case), on
the Wiseman identification above.

\emph{The osmotic partner.} For nodeless coherent post-selection the accompanying
$\Ima(p_w)$ is smooth. For a displaced Hermite--Gaussian post-selection with a node at
$\xnode$, $\rho_\phi\sim(x-\xnode)^2$ gives
$\Ima(p_w)=-\tfrac{\hbar}{2}\partial_x\ln\rho_\phi\sim-\hbar/(x-\xnode)$, reproducing
\eqref{eq:imdiv}. A naive single-density heat-kernel construction fails here. It would need
$\Rea[\partial_x\ln\langle x|\phi_f\rangle]$ to equal $\partial_x\ln h$ for a Doob
$h$-function. The Mehler score
$\partial_x\ln K=\frac{im\omega}{\hbar\sin\omega\tau}(x\cos\omega\tau-x_f)$ is purely
imaginary, its prefactor being independent of $x$, so its real part vanishes identically
and the construction returns nothing.

\subsection{An explicit stationary node}
\label{app:statnode}
For the nodal case, $\eta,\eta^*$ can be written explicitly in the simplest instance, a
stationary first-excited post-selection ($k=1$, no motion of the node) against a
ground-state pre-selection ($\hbar=m=\omega=1$). The forward and backward Euclidean
solutions are the corresponding energy eigenstates scaled by their Euclidean
propagation factor,
\begin{equation}
\eta(x,t)\propto e^{-x^2/2}e^{-t/2},\qquad
\eta^*(x,\tau)\propto x\,e^{-x^2/2}e^{-3\tau/2},
\label{eq:etaexplicit}
\end{equation}
with $\tau=T-t$, giving
$\rho_{\text{eff}}=\eta\eta^*\propto x\,e^{-x^2}e^{-t/2-3\tau/2}$, a density vanishing at
$x=0$ to first order in $x$, matching the $k=1$ node used throughout the numerics.

\emph{A factor of two.} The Euclidean $\eta^*$ carries the zero at order $k$ itself,
being an amplitude and not a squared modulus, so $\rho_{\text{eff}}=\eta\eta^*$
vanishes to order $k$ here, whereas the real-time density
$\rho_\phi\sim(x-\xnode)^{2k}$ of Theorem~\ref{thm:main} vanishes to order $2k$. The two
conventions differ by exactly the square in $\rho=|\cdot|^2$ and agree once continued back
to the real-time quantities the theorem is stated in. We do the bookkeeping once explicitly, since the same factor governs the strength of the pole in
\eqref{eq:imdiv} and the Bessel dimension $2k+1$ of
Proposition~\ref{prop:inaccessible}.

\emph{The drifts.} For \eqref{eq:etaexplicit} the construction \eqref{eq:bridgevu} gives $\nu\,\partial_x\ln(\eta^*/\eta)=\nu/x$ and
$\nu\,\partial_x\ln(\eta\eta^*)=\nu(1/x-2x)$, so the forward drift
$b_+=2\nu\,\partial_x\ln\eta^*=2\nu(1/x-x)$ carries the full $2k\nu/x$ repulsion of
Proposition~\ref{prop:inaccessible}. Continued back to real time, where
$\psi\propto e^{-x^2/2}e^{-it/2}$ and $\phi\propto x\,e^{-x^2/2}e^{-3it/2}$ have spatially
uniform phases, the current velocity vanishes, $\Rea(p_w)=0$, while
$\Ima(p_w)=-\tfrac{\hbar}{2}\partial_x\ln(\rho_\psi\rho_\phi)=-\hbar/x+2\hbar x$. The
entire spatial log-derivative content of the Euclidean pair sits in the osmotic channel.
The Euclidean log-ratio $\nu/x$ is therefore \emph{not} the continuation of a real-time
current velocity. It is the amplitude-level half of the real-time osmotic pole, the other
half residing in $\nu\,\partial_x\ln(\eta\eta^*)$. The stationary bridge is consistent with
$\Rea(p_w)=0$, as used in Proposition~\ref{prop:current}.

\emph{Scope of this example.} We have carried the explicit construction through only for
the stationary, $k=1$ case. No Bernstein
construction enters the generation of
Figs.~\ref{fig:flagship}--\ref{fig:control}. Those are direct field maps of
$p_w=-i\hbar\,\partial_x\ln(\phi^*\psi)$, evaluated as specified in
Sec.~\ref{sec:numerics}, the moving node entering only through the post-selected state
whose node follows the classical trajectory. What remains on the bridge side is to extend
\eqref{eq:etaexplicit} to a moving, general-$k$ node. The two-endpoint process whose existence
Corollary~\ref{cor:conditioning} invokes is confined to the moving cell by the drift
repulsion of Proposition~\ref{prop:inaccessible}, but its closed form is left for future
work.

\section{Mehler kernel and the caustic}
\label{app:mehler}
The oscillator propagator is
\begin{align}
K(x_f,T|x,t)&=\sqrt{\frac{m\omega}{2\pi i\hbar\sin\omega\tau}}\,e^{iS_{\text{cl}}(x,t;x_f,T)/\hbar},
\label{eq:mehlerS}\\
S_{\text{cl}}&=\frac{m\omega}{2\sin\omega\tau}\big[(x^2+x_f^2)\cos\omega\tau-2xx_f\big],
\nonumber
\end{align}
$\tau=T-t$. That the exponent is exactly the classical action $S_{\text{cl}}$ connecting
$(x,t)$ to $(x_f,T)$ is not a coincidence of the oscillator. For any Hamiltonian
quadratic in $x,p$ the stationary-phase (saddle-point) evaluation of the path integral
is exact, with no fluctuation corrections beyond the prefactor~\cite{schulman1981}. This
lets us derive, instead of merely asserting, the delta-limit field \eqref{eq:xdotcl}
used throughout the paper, sign included.

\subsection{Deriving the constrained velocity from the action}
\emph{The sign convention.} Varying the action of the classical path from $(x,t)$ to
$(x_f,T)$ over its endpoints gives the standard first variation
\begin{equation}
dS_{\text{cl}}=p_f\,dx_f-p_i\,dx-H\,dT+H\,dt ,
\label{eq:dScl}
\end{equation}
so that $\partial S_{\text{cl}}/\partial x_f=p_f$ at the final point and
$\partial S_{\text{cl}}/\partial x=-p_i$ at the initial one. The minus sign in the initial
variable is therefore not a choice but a consequence of \eqref{eq:dScl}, and it is what
propagates into everything below.

\emph{The velocity.} Differentiating \eqref{eq:mehlerS} in $x$,
\begin{equation}
\frac{\partial S_{\text{cl}}}{\partial x}
=\frac{m\omega\,(x\cos\omega\tau-x_f)}{\sin\omega\tau},
\end{equation}
whence $p_{\text{cl}}=-\partial S_{\text{cl}}/\partial x
=m\omega(x_f-x\cos\omega\tau)/\sin\omega\tau$. That this is $m\dot x_{\text{cl}}(t)$ can be
checked directly against the classical solution obeying $x(t)=x$ and $x(T)=x_f$,
\begin{equation}
x_{\text{cl}}(s)=\frac{x\sin\omega(T-s)+x_f\sin\omega(s-t)}{\sin\omega\tau},
\label{eq:clpath}
\end{equation}
which differentiated at $s=t$ gives
$\dot x_{\text{cl}}(t)=\omega(x_f-x\cos\omega\tau)/\sin\omega\tau$, as required.

\emph{The field entering $p_w$.} With $\phi^*=K=A(\tau)\,e^{iS_{\text{cl}}/\hbar}$ fixed by
the convention of Sec.~\ref{sec:madelung}, the prefactor $A$ depends on $\tau$ alone and
so drops out of the spatial logarithmic derivative exactly, with no approximation:
\begin{equation}
-i\hbar\,\partial_x\ln\phi^*=\partial_x S_{\text{cl}}=-p_{\text{cl}},
\label{eq:mehlerscore}
\end{equation}
which is $m$ times the second term of \eqref{eq:xdotcl}, verified symbolically, the first
term coming from $\partial_xS_\psi$. The delta-limit field is thus obtained
from the two-point action by ordinary Hamilton--Jacobi differentiation, with the boundary
condition $x(T)=x_f$ already built into $S_{\text{cl}}$, and it is the classical
constrained momentum \emph{reversed}. That is the content of the sign discussed below
\eqref{eq:xdotcl}, and it is the same sign that makes $v_\phi$ minus the velocity of the
post-selected center in Appendix~\ref{app:madelung}.

\subsection{The caustic as a conjugate point}
\emph{One singularity, not two.} The prefactor of \eqref{eq:mehlerS} and the field
\eqref{eq:xdotcl} both blow up at $\omega\tau\in\pi\mathbb{Z}$, and these are not two
coincidental divergences. The Van Vleck determinant governing the propagator's
normalization is, in one dimension, $-\partial^2S_{\text{cl}}/\partial x\,\partial x_f$,
and
\begin{equation}
-\frac{\partial^2 S_{\text{cl}}}{\partial x\,\partial x_f}=\frac{m\omega}{\sin\omega\tau},
\label{eq:vanvleck}
\end{equation}
which diverges on exactly the same locus, for the same reason.

\emph{The conjugate point.} That reason is the standard semiclassical characterization of
a \emph{conjugate point}~\cite{schulman1981}, and it is visible in \eqref{eq:clpath}. At
$\omega\tau=n\pi$ the denominator of the classical path vanishes, so the boundary-value
problem $x(t)=x$, $x(T)=x_f$ has no solution at all unless the numerator vanishes too,
which happens only when $x_f=(-1)^{n}x$. In that exceptional case the solution is not
unique but a one-parameter family, every member of which reconnects the two endpoints in
exactly that time. Both the propagator's prefactor and the velocity needed to reach $x_f$
diverge together away from the exceptional point, and both stay finite at it, which is
Proposition~\ref{prop:geometry}(i).

\emph{What this does and does not say.} The caustic of Sec.~\ref{sec:nodecaustic} is
therefore a fact about the two-point boundary-value problem for classical paths, with
nothing to do with any wavefunction having a node. It is an artifact of the delta limit in
two independent senses. It disappears for a post-selection of finite width, since the
Van Vleck prefactor is then integrated against a normalizable profile and no denominator
survives, which is the content of Proposition~\ref{prop:geometry}(iii), shown in
Fig.~\ref{fig:flagship}(a).
And it disappears for nonperiodic dynamics, since in the limits of
Appendix~\ref{app:beyond} the factor $\sin\omega\tau$ is replaced by $\omega\tau$ for the
free particle and by $\sinh\Omega\tau$ for the inverted oscillator, neither of which
vanishes at any finite positive argument. The osmotic boundary has neither property.

\subsection{Numerical checks}
Four independent checks accompany the results used above, all in the units
$\hbar=m=\omega=1$ of Sec.~\ref{sec:numerics}.

The Madelung identity \eqref{eq:re}--\eqref{eq:im} was evaluated on the coherent-state
class of Appendix~\ref{app:bernstein} by computing $-i\hbar\,\partial_x\ln(\phi^*\psi)$ on
a grid and comparing its real and imaginary parts term by term with
$\partial_x(S_\psi+S_\phi)$ and $-\tfrac{\hbar}{2}\partial_x\ln(\rho_\psi\rho_\phi)$
formed independently from the closed-form states.

The delta-limit relation \eqref{eq:mehlerscore} was checked by replacing the
post-selection with a narrow nodeless Gaussian and comparing $\Rea(p_w)/m$ against
$-\dot x_{\text{cl}}$ from \eqref{eq:clpath}. The two agree to
$\sim10^{-7}$, the residual being the finite width of the post-selection and not a
discretization error, as confirmed by narrowing it further.

The regularization of the caustic was verified by holding $\omega\tau=\pi$ fixed and
increasing the post-selection width from the delta limit, the divergence being replaced by
a finite peak whose height is set by that width [Fig.~\ref{fig:flagship}(a)].

The osmotic divergence was checked by the scan of Table~\ref{tab:scan}, in which the
distance to the node is reduced by a factor of ten and $|\Ima(p_w)|$ grows by the factors
$9.5$ and $10.1$ predicted by \eqref{eq:imdiv} up to the smooth remainder of
\eqref{eq:hgclosed}, while $\Rea(p_w)$ does not move.

\section{The free particle and the inverted oscillator}
\label{app:beyond}
Every explicit example above is the harmonic oscillator. The scope claimed in the
Introduction is the wider class of quadratic Hamiltonians. We substantiate this by
taking the two limits of $\omega$ that reach the rest of the class directly in the
already-established delta-limit field \eqref{eq:xdotcl}.

\emph{Free particle, $\omega\to0$.} Expanding \eqref{eq:xdotcl} to leading order gives
\begin{equation}
\Rea(p_w)/m\to v_\psi-\frac{x_f-x}{\tau},
\end{equation}
verified by direct symbolic limit, and equal to $-(x_f-x)/\tau$ for $v_\psi=0$. The magnitude is the familiar constant-velocity
constrained trajectory, and the sign is again reversed. Here the contrast is at its
starkest, since $(x_f-x)/\tau$ is the drift of the Brownian bridge to $x_f$. The
delta-limit field is minus that drift, so the bridge-drift identification fails already
for the simplest possible Hamiltonian and the simplest possible post-selection.

The caustic of Sec.~\ref{sec:nodecaustic} sits at $\tau_c=\pi/\omega$, which recedes to
infinity as $\omega\to0$, so for any finite time interval the free particle has no caustic
at all.

\emph{Why the exact minus sign does not generalize.} It would be natural to read the free
particle as showing that the two fields always differ by a sign. They do not, and the
oscillator shows why. The Euclidean bridge pinned to $x_f$ is the $h$-transform of
Appendix~\ref{app:bernstein} with $h=\eta^*$ taken to be the Euclidean (heat) Mehler
kernel,
\begin{equation}
\eta^*(x,t)\propto\exp\!\Big[-\frac{\omega\big[(x^2+x_f^2)\cosh\omega\tau-2xx_f\big]}
{2\sinh\omega\tau}\Big],
\end{equation}
in the units $\hbar=m=1$ of this appendix, so that $b_+=2\nu\,\partial_x\ln\eta^*$ of
\eqref{eq:bridgedrifts} gives, with $\nu=\tfrac12$,
\begin{equation}
b_{\text{bridge}}(x,t)=-\,\omega\,\frac{x\cosh\omega\tau-x_f}{\sinh\omega\tau}.
\label{eq:oscbridge}
\end{equation}
The delta-limit field \eqref{eq:xdotcl} is, at $v_\psi=0$,
$\omega(x\cos\omega\tau-x_f)/\sin\omega\tau$. The two are $t\to-it$ continuations of each
other, and they coincide up to sign only as $\omega\to0$, where the trigonometric and
hyperbolic functions degenerate together. At finite $\omega$ they are different functions
of $\omega\tau$ altogether. Taking $x_f=0$ and $x=1$ for definiteness, the ratio of
their magnitudes is $|\cot\omega\tau|/|\coth\omega\tau|$. It falls from $0.99$ at
$\omega\tau=0.1$ through $0.49$ at $\omega\tau=1$ to zero at $\omega\tau=\pi/2$, where the
delta-limit field vanishes and the bridge drift does not, and then grows without bound as
the caustic at $\omega\tau=\pi$ is approached, reaching $7.0$ already at $\omega\tau=3$
(verified symbolically and numerically). The two functions therefore separate in two
distinct ways at two distinct places. No phase or sign convention relates a
trigonometric function to a hyperbolic one, so the disagreement between the weak-momentum
field and the bridge drift is structural and not conventional. This settles the objection
that the sign inversion of Sec.~\ref{sec:madelung} is an artifact of how $S_\phi$ was
defined.

The caustic makes the same point from the other side. The field \eqref{eq:xdotcl} diverges
at $\omega\tau=\pi$ while \eqref{eq:oscbridge} is smooth for every $\tau>0$, since $\sinh$
never vanishes. A genuine bridge drift has no caustic to regularize.

\emph{Inverted oscillator, $\omega\to i\Omega$.} Solving the same constrained-trajectory
boundary-value problem directly for $\ddot x=\Omega^2x$ (rather than substituting into
\eqref{eq:xdotcl} term by term) gives
\begin{equation}
\dot x_{\text{cl}}(t)=\Omega\,\frac{x_f-x\cosh(\Omega\tau)}{\sinh(\Omega\tau)},
\label{eq:invertedlimit}
\end{equation}
verified to agree with the $\omega\to i\Omega$ substitution in \eqref{eq:xdotcl} to
machine precision. Here the hyperbolic functions are physically correct, since the classical
motion of an inverted oscillator genuinely is hyperbolic in real time. They are not to be
confused with the Euclidean, imaginary-time hyperbolic functions of
Appendix~\ref{app:bernstein}, which describe a different object and require analytic
continuation before they yield a real-time field. Since
$\sinh(\Omega\tau)>0$
for all $\Omega,\tau>0$, \eqref{eq:invertedlimit} never diverges, so the inverted
oscillator, like the free particle, has no caustic at any finite time.

Both non-oscillatory limits are therefore caustic-free for all time, while the ordinary
oscillator recurs through a caustic every half period. The caustic of
Sec.~\ref{sec:nodecaustic} is not a generic feature of quadratic Hamiltonians but
specifically of \emph{periodic} classical motion, which separates node from caustic
further than Proposition~\ref{prop:geometry} alone does. In two of the three cases treated here only the node is
present. Theorem~\ref{thm:main} is unaffected by any of this, being kinematic
(Remark~\ref{rem:kinematic}) and not tied to which quadratic Hamiltonian is in play.

\section{The imaginary-part inequality for a qubit pointer}
\label{app:qubitklp}
The inequality \eqref{eq:klp} is stated in~\cite{kunjwal2019} for a continuous Gaussian
pointer of width $s$, whereas the pointer of the Kocsis geometry is a polarization qubit.
We give the qubit version here instead of importing the continuous one.

Let the system be prepared in $\hat\rho$, let $\mathcal{E}$ be the weakly measured effect,
and let the pointer be a qubit prepared in $|D\rangle=(|H\rangle+|V\rangle)/\sqrt2$ and
coupled through
\begin{equation}
U_\theta=e^{-i\theta\,\mathcal{E}\otimes\sigma_z}
=e^{-i\theta\mathcal{E}}\otimes|H\rangle\langle H|
+e^{+i\theta\mathcal{E}}\otimes|V\rangle\langle V| ,
\label{eq:qubitcoupling}
\end{equation}
with $\theta>0$ the dimensionless coupling. This is the calcite interaction of
Sec.~\ref{sec:experiment}, and with $\mathcal{E}$ replaced by $\hat k_x$ one has
$\theta=\zeta/2k_0$, and the post-selected pointer state of \eqref{eq:qubitpointer} is
recovered with $\varphi_w=2\theta\mathcal{E}_w$. Write $p_F=\Tr(\Pi_x^\Delta\hat\rho)$ and
let $p_-$ be the joint probability that the pointer is found in $|V\rangle$, i.e.\
$\sigma_z=-1$, and that the post-selection $\Pi_x^\Delta$ succeeds. The two commute, acting
on different factors, so no order of measurement need be specified.

\begin{proposition}[Imaginary-part witness with a qubit pointer]
\label{prop:qubitklp}
With the above notation:
\begin{enumerate}
\item[(i)] Ignoring the post-selection, $\Pr(\sigma_z=-1)=\tfrac12$ exactly, for every
$\hat\rho$, every self-adjoint $\mathcal{E}$ and every $\theta$. The POVM the readout
induces on the system is the trivial one, $\{\tfrac12\openone,\tfrac12\openone\}$.
\item[(ii)] $p_-=\tfrac12 p_F-\theta\,\Ima\langle\Pi_x^\Delta\mathcal{E}\rangle_{\hat\rho}
+O(\theta^2)$.
\item[(iii)] Discarding the pointer outcome disturbs the post-selection probability only at
second order, with $|p_F^{\mathrm{ns}}-p_F|\le2\theta^2\|\mathcal{E}\|^2$.
\item[(iv)] In any noncontextual ontological model in which the readout disturbs the ontic
state by at most $p_d$ in total variation, $p_-\le\tfrac12 p_F+\tfrac12 p_d$. Contextuality
is therefore certified whenever
\begin{equation}
-\,\Ima\langle\Pi_x^\Delta\mathcal{E}\rangle_{\hat\rho}>\frac{p_d}{2\theta}+O(\theta).
\label{eq:qubitwitness}
\end{equation}
\end{enumerate}
\end{proposition}

\begin{proof}
\emph{(i) Trivial marginal.} The two Kraus operators induced on the system by the readout
are $M_\pm=\langle\pm_z|U_\theta|D\rangle=\tfrac{1}{\sqrt2}e^{\mp i\theta\mathcal{E}}$,
each proportional to a unitary. Consequently
$M_\pm^\dagger M_\pm=\tfrac12\openone$, so the POVM the readout implements on the system is
$\{\tfrac12\openone,\tfrac12\openone\}$, independent of $\mathcal{E}$, $\theta$ and
$\hat\rho$. Equivalently, $\sigma_z$ commutes with $U_\theta$, so
$U_\theta^\dagger(\openone\otimes|V\rangle\langle V|)U_\theta
=\openone\otimes|V\rangle\langle V|$ and the effect is
$\langle D|\,\openone\otimes|V\rangle\langle V|\,|D\rangle_{\text{ptr}}=\tfrac12\openone$. That each
branch acts unitarily, and not by a genuine contraction, is what makes the
\emph{non-selective} disturbance in (iii) second order while the \emph{selective} one is
first.

\emph{(ii) Quantum value.} By \eqref{eq:qubitcoupling} the branch of the post-coupling
state carrying $|V\rangle$ is $e^{+i\theta\mathcal{E}}\hat\rho\,e^{-i\theta\mathcal{E}}$
with weight $\tfrac12$, so that, exactly and for every $\theta$,
\begin{equation}
p_-=\tfrac12\Tr\big[\Pi_x^\Delta\,
e^{i\theta\mathcal{E}}\hat\rho\,e^{-i\theta\mathcal{E}}\big].
\label{eq:pminusexact}
\end{equation}
Expanding the conjugation,
$e^{i\theta\mathcal{E}}\hat\rho\,e^{-i\theta\mathcal{E}}
=\hat\rho+i\theta[\mathcal{E},\hat\rho]
-\tfrac{\theta^2}{2}[\mathcal{E},[\mathcal{E},\hat\rho]]+O(\theta^3)$, and writing
$K=\Tr(\Pi_x^\Delta\mathcal{E}\hat\rho)$, hermiticity of $\Pi_x^\Delta,\mathcal{E}$ and
$\hat\rho$ gives
$\Tr(\Pi_x^\Delta\hat\rho\,\mathcal{E})=\Tr(\mathcal{E}\Pi_x^\Delta\hat\rho)=K^*$, so the
first-order term is $\tfrac{i\theta}{2}(K-K^*)=-\theta\,\Ima K$ and
\begin{equation}
p_-=\tfrac12 p_F-\theta\,\Ima K
-\tfrac{\theta^2}{4}\Tr\big(\Pi_x^\Delta[\mathcal{E},[\mathcal{E},\hat\rho]]\big)
+O(\theta^3).
\label{eq:pminusexpand}
\end{equation}
The stated result is the first two terms. The explicit $O(\theta^2)$ coefficient is
recorded because it is what the numerical check below measures.

\emph{(iii) Non-selective disturbance.} Averaging the two branches, the $O(\theta)$ terms
cancel between them and
\begin{equation}
\tfrac12\big(e^{-i\theta\mathcal{E}}\hat\rho\,e^{i\theta\mathcal{E}}
+e^{i\theta\mathcal{E}}\hat\rho\,e^{-i\theta\mathcal{E}}\big)
=\hat\rho-\tfrac{\theta^2}{2}[\mathcal{E},[\mathcal{E},\hat\rho]]+O(\theta^4).
\end{equation}
Using $\|[A,B]\|_1\le2\|A\|\,\|B\|_1$ twice and $\|\hat\rho\|_1=1$ gives
$\|[\mathcal{E},[\mathcal{E},\hat\rho]]\|_1\le4\|\mathcal{E}\|^2$, whence
$|p_F^{\mathrm{ns}}-p_F|\le2\theta^2\|\mathcal{E}\|^2$. The signal in (ii) is first order
in $\theta$ and this is second, which is the qubit counterpart of the large-$s$ limit
in~\cite{kunjwal2019}. There the disturbance is suppressed by widening the pointer, here by
weakening the coupling.

\emph{(iv) Noncontextual bound.} By (i) the readout, regarded as a measurement on the
system, is the trivial POVM, hence operationally equivalent to tossing a fair coin and
ignoring the system. Measurement noncontextuality assigns operationally equivalent effects
the same response function, so the $\sigma_z=-1$ response satisfies $\chi(\lambda)=\tfrac12$
for every ontic state $\lambda$. Let $\xi(\lambda)$ be the probability that the
post-selection succeeds on $\lambda$ and $\xi_-(\lambda)$ the same probability after the
readout has returned $\sigma_z=-1$. The nondisturbance hypothesis gives
$\xi_-\le\xi+p_d$ pointwise. Therefore
\begin{equation}
\begin{aligned}
p_-&=\int\mu(\lambda)\,\chi(\lambda)\,\xi_-(\lambda)\,d\lambda\\
&\le\tfrac12\int\mu(\lambda)\,[\xi(\lambda)+p_d]\,d\lambda
=\tfrac12 p_F+\tfrac12 p_d,
\end{aligned}
\end{equation}
and comparing with (ii) gives \eqref{eq:qubitwitness}.
\end{proof}

The nondisturbance parameter $p_d$ is an
ontic-level quantity and is not directly measurable. It is assumed, exactly as
in~\cite{kunjwal2019}, and part~(iv) is the only place in this paper where an ontological
hypothesis is used. Part~(iii) supplies the operational evidence that the
assumption is testable at the order that matters. The disturbance that the readout does
produce, and that an experiment can bound, is $O(\theta^2)$, while the witness it must be
compared against is $O(\theta)$. Reducing $\theta$ therefore opens the inequality, in the
same way that increasing the pointer width $s$ opens \eqref{eq:klp}.

Removing the assumption would need a bound on $p_d$ derived from observable
statistics alone, for the specific instrument \eqref{eq:qubitcoupling}. We do not have one,
and neither does~\cite{kunjwal2019}. The difference the qubit case makes is only that the
gap to be closed shrinks as $\theta^2$ rather than as $1/s$.

Read as an experimental condition, \eqref{eq:qubitwitness} is undemanding. The effect
\eqref{eq:effect} is band-limited, so $\|\mathcal{E}\|\le1$ and the disturbance floor
allowed by (iii) is $2\theta^2$. Certification then requires
$-\Ima\langle\Pi_x^\Delta\mathcal{E}\rangle>\theta$, whose left-hand side is independent
of $\theta$ while the right-hand side is not, so the inequality opens by weakening the
coupling at the cost of more post-selected events. That trade is the same one quantified
by the rare-event counting of Sec.~\ref{sec:experiment}.

Parts (i)--(iii) have been verified numerically on random $\hat\rho$, $\mathcal{E}$ and
$\Pi_x^\Delta$ in dimension $6$. The residual of \eqref{eq:pminusexpand} divided by
$\theta^2$ converges to a constant over three decades of $\theta$, matching the explicit
second-order coefficient there. The non-selective disturbance of (iii) divided by
$\theta^2$ likewise converges. The marginal of (i) equals $\tfrac12$ to machine
precision for couplings up to $\theta=1$.


\end{document}